\documentclass[aps,pra,letterpaper,superscriptaddress,notitlepage,onrcolumn,nofootinbib]{revtex4-2} 
\usepackage{natbib}
\usepackage{amsfonts}
\usepackage{amsmath}
\usepackage{amssymb}
\usepackage{amsthm}
\usepackage{mathtools}
\usepackage{enumerate}
\usepackage{bm}
\usepackage{dsfont}
\usepackage{graphicx}
\usepackage{nicefrac}
\usepackage{color}
\usepackage{soul}
\usepackage[all]{xy}
\usepackage{hyperref}
\usepackage{bbm}
\usepackage{tikz}
\usepackage{tikz-cd}
\usepackage{cancel}
\usepackage{xcolor}
\usepackage{booktabs}
\usepackage{multirow}
\usetikzlibrary{positioning}
\usetikzlibrary{arrows}
\usepackage{MnSymbol}
\usepackage[normalem]{ulem}

\usepackage[paperwidth=210mm,paperheight=297mm,centering,hmargin=2.1cm,vmargin=2.1cm]{geometry}

\newcommand{\calC}{\mathcal C}

\newcommand{\calS}{\mathcal S}

\newcommand{\calG}{\mathcal G}

\newtheorem{theorem}{Theorem}

\newtheorem{corollary}{Corollary}
\newtheorem{lemma}{Lemma}
\newtheorem{definition}{Definition}
\newtheorem*{question*}{Question}

\newcommand{\ket}[1]{\ensuremath{\vert#1\rangle}}

\newcommand{\zz}{\mathbb Z}
\newcommand{\rr}{\mathbb R}

\newcommand\mc[1]{\mathcal{#1}}

\graphicspath{{./figures/}}

\begin{document}
 \bibliographystyle{plainnat}
\title{Hierarchies of resources for measurement-based quantum computation}
\author{Markus Frembs}%
\affiliation{Centre for Quantum Dynamics, Griffith University, Gold Coast, QLD 4222, Australia}%
\author{Sam Roberts}%
\affiliation{Centre for Engineered Quantum Systems, School of Physics, The University of Sydney, Sydney, NSW 2006, Australia}%
\author{Earl Campbell}%
\affiliation{Department of Physics \& Astronomy, University of Sheffield, Sheffield, S3 7RH, United Kingdom}
\affiliation{Riverlane, Cambridge  CB2 3BZ, United Kingdom}
\author{Stephen Bartlett}%
\affiliation{Centre for Engineered Quantum Systems, School of Physics, The University of Sydney, Sydney, NSW 2006, Australia}%
\begin{abstract}
For certain restricted computational tasks, quantum mechanics provides a provable advantage over any possible classical implementation.  Several of these results have been proven using the framework of measurement-based quantum computation (MBQC), where non-locality and more generally contextuality have been identified as necessary resources for certain quantum computations.  
Here, we consider the computational power of MBQC in more detail by refining its resource requirements, both on the allowed operations and the number of accessible qubits. More precisely, we identify which Boolean functions can be computed in non-adaptive MBQC, with local operations contained within a finite level in the Clifford hierarchy. Moreover, for non-adaptive MBQC restricted to certain subtheories such as stabiliser MBQC, we compute the minimal number of qubits required to compute a given Boolean function. Our results point towards hierarchies of resources that more sharply characterise the power of MBQC beyond the binary of contextuality vs non-contextuality.

\end{abstract}

\maketitle

\section{Introduction}

Quantum computation promises many advantages over classical computation, including the ability to efficiently solve certain problems, such as factoring, where no efficient classical algorithms are currently known. What drives this quantum advantage?

Contextuality offers a potential answer to this question, as it has been found to be an important resource for quantum computation in a variety of settings~\cite{Raussendorf2009,howard2014contextuality,bermejo2017contextuality,delfosse2015wigner,raussendorf2017contextuality,karanjai2018contextuality,raussendorf2016cohomological,okay2017topological,veitch2014resource,mansfield2018quantum,pashayan2015estimating,deSilva2018,FrembsRobertsBartlett2018,shahandeh2021quantum}. Roughly speaking, contextuality is the impossibility of assigning pre-determined outcomes to all potential measurements of a quantum system in a way that is independent of other, simultaneously performed measurements \cite{KochenSpecker1967}. Contextuality is a common notion of non-classicality. Notably, contextuality plays a central role in a recent seminal result showing a provable quantum advantage for a class of shallow quantum circuits over their classical counterparts \cite{BravyiGossetKoenig2018} (later extended to the noisy setting in Ref.~\cite{bravyi2020quantum}). While the class of problems solvable with such circuits is not motivated by practical applications, it provides a proof of principle that quantum advantages over classical computation are possible, and highlights quantum contextuality as a key resource.

Despite this evidence for the role of contextuality as a resource for quantum advantage, a finer characterisation of this resource is largely missing.  We address this problem by asking a related question: how non-classical is quantum computation? This is similar to the study of the extent to which quantum mechanics violates certain Bell inequalities, yet with an explicit emphasis on computation and computationally relevant resource constraints.

In this paper, we study the computability of Boolean functions in the framework of measurement-based quantum computation (MBQC)~\cite{raussendorf2001one,RaussendorfBrowneBriegel2003,briegel2009measurement}, observing that many of the relevant results in the literature including Refs.~\cite{BravyiGossetKoenig2018,AndersBrowne2009} are readily and naturally formulated within the measurement-based framework. For simplicity, we focus on non-adaptive MBQC with linear side-processing, where contextuality provides the sharpest known separation between classical and quantum computation \cite{Raussendorf2013,FrembsRobertsBartlett2018}. We outline this setup in Sec.~\ref{sec: setting} below.

Within this setting, we further consider the interplay between the following two resource aspects: the amount of magic (non-Clifford operations, see Sec.~\ref{sec: stabiliser subtheory}) necessary and the number of qubits required for the computation of a given Boolean function. Already in this limited framework, the classification of Boolean functions under these resources points towards a rich structure beyond the classical paradigm. We summarise our main results and provide an overview to the structure of the paper in Sec.~\ref{secSummaryOfResults}.

\subsection{The setting}\label{sec: setting}
 
In this section, we define our restricted framework of MBQC. A MBQC consists of a correlated quantum resource state, and a control computer with restricted computational power. The quantum resource state consists of $N$ local subsystems---or parties---each of which consists of a qubit and measurement device that exchanges classical information with the control computer once. The control computer is responsible for selecting the measurement settings for each local subsystem, and for processing the measurement outcomes into useful computational output. Importantly, the power of the control computer is limited: we consider control computers that can only compute linear functions, and as such are not even classically universal.\footnote{The restriction to linear side-processing greatly simplifies the analysis of contextuality as a resource in MBQC. While nonlinear side-processing is not required for universal MBQC, one may consider relaxing this restriction in future studies in order to quantify any advantage of (MB)QC over universal classical computation in practical settings.} This notion of MBQC is known as $l2$-MBQC (where the $l2$ stands for mod-2 linear side-processing) and is based on the model of Anders and Browne~\cite{AndersBrowne2009}. The following definition is based on Refs.~\cite{Raussendorf2013,FrembsRobertsBartlett2018}. (See Ref.~\cite{okay2017topological} for a more general notion of MBQC.)

\begin{definition}\label{defldMBQC}
    A $l2$-MBQC with classical input $\mathbf{i} \in \zz_2^n$ and classical output $o \in \zz_2$ consists of $N$ qubit subsystems, jointly prepared in the state $\ket{\psi}$, each of which receives an input $c_k(\mathbf{i}) \in \zz_2$ from the control computer, performs a measurement $M_k(c_k(\mathbf{i}))$, and returns a measurement outcome $m_k \in \zz_2$, for $k = 1,\ldots, N$.\footnote{Throughout, we will use boldface for vectors.} The inputs and computational output satisfy the following conditions:
    \begin{enumerate}
    \item The computational output $o \in \zz_2$ is a linear function of the local measurement outcomes $\mathbf{m} = (m_1,\cdots,m_N)^\intercal\in \zz_2^N$,
    \begin{equation*}\label{eq: MBQC linear post-processing}
        o = \sum_{k=1}^N m_k + m_0 \quad \mathrm{mod}\ 2\; ,
    \end{equation*}
    for some $m_0 \in \zz_2$.  
    \item Local measurements $M_k(c_k)$ have eigenvalues $(-1)^{m_k}$.  
    The measurement settings $\mathbf{c} = (c_1,\cdots,c_N)^\intercal \in \zz_2^N$ are linear functions of the classical input $\mathbf{i} = (i_1,\cdots,i_n)^\intercal\in \zz_2^n$ and the measurement outcomes $\mathbf{m}$ via
    \begin{equation}\label{eq: MBQC linear pre-processing}
        \mathbf{c} = T\mathbf{m} + P\mathbf{i} \quad \mathrm{mod}\ 2,
    \end{equation}
    for some $T \in \text{Mat}(N \times N, \zz_2)$ and $P \in \text{Mat}(N \times n,\zz_2)$.
    \item For a suitable ordering of the parties $1,\ldots,N$ the matrix $T$ in Eq.~(\ref{eq: MBQC linear pre-processing}) is lower triangular with vanishing diagonal. If $T = 0$ the $l2$-MBQC is called non-adaptive.
    \end{enumerate}
\end{definition}

\begin{figure}[h]%
	\centering
	\includegraphics[width=0.65\linewidth]{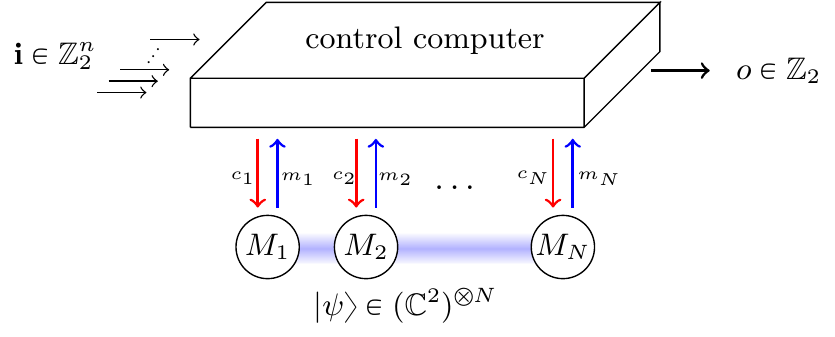}
	\caption{The schematic setup of an l2-MBQC defined in Def.~\ref{defldMBQC}, from Ref.~\cite{FrembsRobertsBartlett2018}. 
	For each qubit (indexed by $k$) of the resource state $\ket{\psi}$, the control computer determines the measurement settings $c_k$ as a linear function of the inputs $\mathbf{i}\in \zz_2^n$ and any previous measurement outcomes $m_1, \ldots, m_{k-1}$. The output $o \in \zz_2$ is evaluated by the control computer as the parity of the measurement outcomes.}
\end{figure}

Whenever the output of the computer is a deterministic function of the inputs we have $o = o(\mathbf{i})$ for $\mathbf{i} \in \zz_2^n$. We say the $l2$-MBQC is deterministic in this case. More generally, in the non-deterministic (probabilistic) case every input specifies a probability distribution over the outputs. We will mostly restrict ourselves to deterministic $l2$-MBQC (with the exception of Thm.~\ref{thm: stabiliser success probability}). Moreover, we will focus on the non-adaptive case. The latter is a natural restriction for the study of contextuality (nonlocality) as a resource in MBQC \cite{HobanCampbell2011}, since adaptivity generally allows to reproduce any nonlocal correlations (see also Remark 1 in \cite{Raussendorf2013}). Nevertheless, more flexible restrictions on adaptivity can still lead to interesting classes of  algorithms such as shallow circuits in \cite{BravyiGossetKoenig2018}. We briefly discuss the adaptive case in App.~\ref{secAdaptivity}. Finally, we note that Def.~\ref{defldMBQC} is readily generalised to qudit systems, but requires care in the definitions of the higher-dimensional measurements allowed within the framework. Many of our results generalise to qudit systems of prime dimension, yet additional technicalities arise; to simplify presentation we only consider the qubit case in the main body of the text.

\subsection{The stabiliser subtheory}\label{sec: stabiliser subtheory}
We denote the group of Pauli operators on $N$ qubits as $\mathcal{P}_N$. 
Throughout, we label the local computational basis states as  $|q\rangle$ for $q \in \{0,1\}$. 
An important class of operators is given by the Clifford hierarchy.

\begin{definition}\label{def: Clifford hierarchy}
    The Clifford hierarchy on $N$ qubits is defined recursively by setting $\mathcal{C}_N^{1} = \mathcal{P}_N$, and letting the $k$'th level $\mathcal{C}_N^{k}$ be given by
	\begin{equation}
	    \mathcal{C}_N^{k} = \{ U \in \mathcal{U}((\mathbb{C}^2)^{\otimes N}) ~|~ UPU^{\dagger} \in \mathcal{C}_N^{k-1} ~\forall P \in \mathcal{P}_N \}.
	\end{equation}
\end{definition}

Notably, the second level $\mathcal{C}_N^{2}$ is the normaliser of the Pauli group and is known as the \emph{Clifford group}. Any state that can be obtained by applying a gate from the Clifford group to a computational basis state is known as a \emph{stabiliser state}. Note that in the setting of the Clifford hierarchy, it is natural to model the classical control in Def.~\ref{defldMBQC} in the form of unitary conjugation on some fiducial measurement setting.
\begin{definition}\label{defLevelDMBQC}
	We say a MBQC belongs to level-$D$ if the local measurement settings are of the form of 
	\begin{equation}\label{eq: ld-MBQC operators by unitary conjugation}
	    M_k(c_k) = U_k(c_k) M_k(0) U_k^{-1}(c_k)\; ,
	\end{equation}
	where $M_k(0)\in \mathcal{P}_1$ is some fiducial measurement setting, $U_k(c_k) \in \mathcal{C}_1^{D}$, and where the resource state is a stabiliser state.
\end{definition}
When the $l2$-MBQC belongs to level-2, the MBQC belongs to the stabiliser subtheory, and is classically efficiently simulable by the Gottesman-Knill theorem~\cite{gottesman1998heisenberg,aaronson2004improved}. Level-3 MBQCs are universal for quantum computation (in the adaptive case), with the scheme based on cluster states~\cite{RaussendorfBrowneBriegel2003} being a well known example. The restriction on resource states being stabiliser states is without loss of generality -- one can additionally allow resource states that are obtained by applying a $D$'th level gate to a stabiliser state, in close analogy with the paradigm of stabiliser quantum computing supplemented by magic state injection. 

In the context of MBQC, it is convenient to express the output of the computation in terms of a polynomial. Namely, every Boolean function $f\,:\zz_2^{n}\longrightarrow \zz_2$ is given by a polynomial from the ring $\zz_2[{x}_{1},\,\ldots ,\,{x}_{n}]$ in $n$ variables ${x}_{1},\,\ldots ,\,{x}_{n}\in \zz_2$. This representation is known as the algebraic normal form.

\subsection{Summary of results}\label{secSummaryOfResults}

In this paper, we study the computability of Boolean functions in non-adaptive $l2$-MBQC under various resource constraints. Below, we summarize our main results, and outline the structure of the rest of the paper.

\textbf{Contextuality.} We begin by recalling that any Boolean function $f: \zz_2^n \rightarrow \zz_2$ can be computed within non-adaptive, deterministic $l2$-MBQC \cite{HobanCampbell2011} (see Thm.~\ref{thm: ld-MBQC universal} in Sec.~\ref{secCompleteness}). In the classical setting, only linear functions are computable. Thus, nonlinearity indicates the presence of quantumness in the form of contextuality \cite{Raussendorf2013,FrembsRobertsBartlett2018}. The proof of this result relies on operators outside the Clifford group, i.e., outside the second level in the Clifford hierarchy; moreover, it generally requires an exponential (in the degree of $f$, expressed as a polynomial) number of qubits. This suggests a finer classification in terms of the Clifford hierarchy, which we present in Sec.~\ref{secstabiliserMBQC}, and the number of qubits (`qubit count') required to implement a given Boolean function in the non-adaptive case, presented in Sec.~\ref{SecGeneralResourceRequirements}.\footnote{In the adaptive case, one must also consider the time required to implement a given function, which we briefly address in App.~\ref{secAdaptivity}.} A natural starting point for these considerations is the stabiliser sub-theory, where resource states are stabiliser states and operators are restricted to the second level in the Clifford hierarchy.

\textbf{Stabiliser theory.} In the case of $l2$-MBQCs belonging to level-2 (i.e., stabiliser MBQCs), we show the computable functions (in non-adaptive MBQC) to be heavily restricted: in the deterministic case, only quadratic functions can be computed (see Thm.~\ref{thm:stabdeterm}), while in the probabilistic case, the success probability (see Def.~\ref{def: success probability}) to compute a given Boolean function is bounded by its non-quadraticity (see Def.~\ref{def: non-quadraticity}), i.e., the Hamming distance to the nearest quadratic function (see Thm.~\ref{thm: stabiliser success probability}). These results are presented in Sec.~\ref{sec: Stabiliser formalism and quadratic Boolean functions}.

Moreover, we find that in the deterministic case, a quadratic function can be implemented using $\mathrm{rk}(f)+1$ qubits only, where $\mathrm{rk}(f)$ denotes the rank of the matrix corresponding to the quadratic terms of $f$ (see Thm.~\ref{thm: resource of quadratic Boolean functions}).

\textbf{Clifford hierarchy.} Despite being non-classical (contextual), the above mentioned results (Thm.~\ref{thm:stabdeterm} and Thm.~\ref{thm: stabiliser success probability}) show that computation within non-adaptive stabiliser $l2$-MBQC is limited.\footnote{Note that for $d$ odd prime, the stabiliser formalism is in fact non-contextual \cite{Gross2006}. At least in this case, we can take it as the lowest level of such a hierarchy.} A natural way to extend the stabiliser case is via the Clifford hierarchy. In Sec.~\ref{sec: Beyond quadratic functions}, we consider what non-Clifford resources are required to implement a given Boolean function within $l2$-MBQC. The main result of this section, Thm.~\ref{thm:CliffHierNec} shows that operations from the $D$'th level in the Clifford hierarchy are required whenever a non-adaptive, deterministic $l2$-MBQC computes a polynomial of degree $D$.

\textbf{Qubit count.} While we can compute the minimal number of qubits in the stabiliser case, i.e., for quadratic functions (see Thm.~\ref{thm: resource of quadratic Boolean functions} in Sec.~\ref{sec: Quadratic Boolean functions}), generalizing this result beyond the stabiliser case is challenging. In Sec.~\ref{sec: Optimal representation of functions in non-adaptive l2-MBQC}, we consider an approach based on GHZ states, which (by the proof of Thm.~\ref{thm: ld-MBQC universal}) provide a universal resource for function computation in non-adaptive, deterministic $l2$-MBQC.\footnote{Note however, that GHZ states are not universal for MBQC in general.} We characterise the number of qubits required to compute an arbitrary Boolean function in terms of the minimal number of Fourier components (see Thm.~\ref{thm: zero term reduction}). Similar optimisation problems arise in circuit synthesis \cite{AmyMosca2016,seroussi1983maximum,heyfron2018efficient,Kissenger,heyfron2019quantum}.

In addition, we employ the discrete Fourier transform to obtain upper bounds on the qubit count for certain highly symmetric functions, which turn out to be optimal in some cases, e.g. for $\delta$-functions Cor.~\ref{cor: optimality delta-function}. As an immediate consequence, we conclude that the number of qubits required to implement a Boolean function $f$ in non-adaptive, deterministic $l2$-MBQC is far from monotonic in the degree of $f$ (see Cor.~\ref{cor: mismatch degree vs qubit count}), thus further hinting at a rich substructure of contextuality beyond the results in Refs.~\cite{raussendorf2016cohomological,FrembsRobertsBartlett2018}. 

Finally, we discuss possible avenues towards related and future research in Sec.~\ref{sec: discussion}.

\section{Every Boolean function has a representation as contextual MBQC}\label{secCompleteness}

In this section we prove Theorem~\ref{thm: ld-MBQC universal}, that non-adaptive $l2$-MBQC is complete. That is, for any function $f:\zz_2^n \rightarrow \zz_2$ there exists an $l2$-MBQC with output function $o(\mathbf{i})=f(\mathbf{i})$ for all inputs $\mathbf{i} \in \zz_2^n$. This is in sharp contrast to the classical regime, which is restricted to linearity---nonlinear computation is an indicator of quantum contextuality \cite{Raussendorf2013,FrembsRobertsBartlett2018}. The proof strategy is to first construct $l2$-MBQCs that compute the $n$-dimensional $\delta$-function. Linearly composing the output of many such parallel $l2$-MBQCs can then be used to compute any function. In fact, our proof is easily generalised to qudits of prime dimension (see App.~\ref{sec: Proof of lm: qudit delta function}).\\

We begin by defining the resource state and the measurement operators relevant for this construction. We take the resource state to be given by the $N$-qubit GHZ state 
\begin{equation}\label{eq: GHZ resource state}
    |\psi\rangle = \frac{1}{\sqrt{2}} (|0\rangle^{\otimes N} + |1\rangle^{\otimes N})\; .
\end{equation}
This is a mild restriction, since the GHZ state in Eq.~(\ref{eq: GHZ resource state}) will prove to be a universal resource for non-adaptive, deterministic $l2$-MBQC in Thm.~\ref{thm: ld-MBQC universal} below (see also \cite{HobanCampbell2011,WernerWolf2001}). More generally, in Sec.~\ref{secstabiliserMBQC} we will define a hierarchy for $l2$-MBQC by restricting the allowed operations to certain levels in the Clifford hierarchy and the resource state to a stabiliser state (see Def.~\ref{defLevelDMBQC}). Note also that the GHZ state is a stabiliser state. Finally, in Sec.~\ref{SecGeneralResourceRequirements} we will analyse the qubit count for $l2$-MBQC with a GHZ resource state.

Next, recall from Def.~\ref{defldMBQC} that each party performs one of two measurements $M_k(c_k)$ determined by a single input $c_k \in \zz_2$. Moreover, we require that $M_k$ has (non-degenerate) eigenvalues $(-1)^q$, $q \in \zz_2$, i.e., $M_k^2 = 1$. We define the following canonical measurement operators
\begin{equation}\label{eq: l2-MBQC operators X(f)}
    X(\theta)|q\rangle = \theta^{1-2q}|q\oplus 1\rangle = e^{i\pi(1-2q)\vartheta}|q\oplus 1\rangle\; ,
\end{equation}
where $\theta = e^{i\pi \vartheta}$. 
In matrix (gate) representation, these operators take the form
\begin{equation}\label{eq: l2-MBQC operators X(theta,f)}
    X(\theta) = \left( \begin{array}{cc}
       0 & \theta^*  \\
    \theta & 0  
    \end{array} \right)\; .
\end{equation}

The inputs $c_k$ to the measurement devices thus specify $M_k(c_k) = X_k(\theta(c_k))$ and are themselves determined in a linear way from the computational input $\mathbf{i}\in \zz_2^n$ and other measurement outcomes $m_k \in \zz_2$ according to the general setup in Def.~\ref{defldMBQC}. (Note that in the non-adaptive case, $c_k = c_k(\mathbf{i})$ is a linear functions of the inputs only.)

The output function of the $l2$-MBQC, $o(\mathbf{i}) = \oplus_{k=1}^N m_k$, arises as the parity of the individual measurement outcomes on local qubits. The resource state $|\psi\rangle$ is a $+1$-parity eigenstate of the operator $\otimes_{k=1}^N X_k(0)$. On the other hand, we can easily construct operators for which $|\psi\rangle$ is a $(-1)$-parity eigenstate. For instance, consider the prototypical Anders-Browne $3$-qubit example, where $M_k(0) = X_k(0) = X_k$ and $M_k(1) = X_k(\sqrt{-1}) = Y_k$.
Note that this choice of local measurements solves the following set of linear equations $\sum_{k=1}^3 c_k(i_1,i_2) \cdot \vartheta_k = o(i_1,i_2)$, where $\theta_k = e^{i\pi \vartheta_k}$, $\vartheta_k = \frac{1}{2}$, and $c_1(i_1,i_2) = i_1$, $c_2(i_1,i_2) = i_2$, $c_3(i_1,i_2) = i_1 \oplus i_2$, and $o(i_1,i_2) = i_1i_2 \oplus i_1 \oplus i_2$.\\

In fact, this example is representative of the general case. More precisely, for deterministic $l2$-MBQC the computation can be expressed in terms of the phase parameters in the local measurement operators of Eq.~(\ref{eq: l2-MBQC operators X(f)}).

\begin{theorem}\label{thm: phase relations ld-MBQC}
	In every non-adaptive, deterministic $l2$-MBQC with a GHZ resource state, the output function $o: \zz_2^n \rightarrow \zz_2$ arises from the phase relations between local measurement operators in Eq.~(\ref{eq: l2-MBQC operators X(f)}),
	\begin{equation}\label{eq: phase relations}
	    o(\mathbf{i}) = \sum_{k=1}^N c_k(\mathbf{i})\vartheta_k \pmod 2 \quad \forall \mathbf{i} \in \zz_2^n\; .
	\end{equation}
\end{theorem}

\begin{proof}
	We give the proof in App.~\ref{app: proof of Thm 1}.
\end{proof}

Finding an implementation to compute $o$ as a $l2$-MBQC thus reduces to finding a set of (linear) functions $c_k$, which satisfies the required parity conditions in Eq.~(\ref{eq: phase relations}). 

We first construct an $l2$-MBQC that computes the $n$-dimensional $\delta$-function $\delta : \zz_2^n \rightarrow \zz_2$ defined by
\begin{align}\label{eq: n-dimensional delta function}
    \delta(\mathbf{i}) &:= \begin{cases} 1 &\mathrm{if\ } \mathbf{i}=\mathbf{0} \\ 0 &\mathrm{elsewhere} \end{cases}\; .
\end{align}

We remark that the $\delta$ function is an important and ubiquitous function -- up to linear pre- and post-processing it is equivalent to the $n$-bit AND function. We have the following lemma.

\begin{lemma}\label{lem: qubit delta function}
	The $n$-dimensional $\delta$-function can be implemented on $N = 2^n-1$ qubits within non-adaptive, deterministic $l2$-MBQC.
\end{lemma}

\begin{proof}
	We prove this in App.~\ref{app: proof of Lm 1} by giving an explicit measurement scheme acting on a GHZ state.
\end{proof}

We remark that a similar result has previously been obtained in Ref.~\cite{HobanCampbell2011}. Here, we gave a constructive proof in terms of the operators in Eq.~(\ref{eq: l2-MBQC operators X(theta,f)}). Moreover, our technique generalises to qudits of prime dimension (for details, see App.~\ref{sec: Proof of lm: qudit delta function}).

In particular, we note that Lm.~\ref{lem: qubit delta function} recovers the main example of Anders and Browne~\cite{AndersBrowne2009} (up to linear side-processing) for $n=2$ with $\theta_k = e^{i\pi \vartheta_k}, \vartheta_k = \frac{1}{2}$, such that $M(0) = X$ and $M(1) = Y$.\\

The $n$-dimensional $\delta$-function along with linear side-processing is sufficient to allow for the evaluation of arbitrary functions. In particular, one can decompose any function into a linear combination of delta functions, each of which admits an $l2$-MBQC. The outputs of these $l2$-MBQCs can be linearly combined to give the desired output, as in the following theorem.

\begin{theorem}\label{thm: ld-MBQC universal}
	For any Boolean function $f: \zz_2^n \rightarrow \zz_2$ there exists a non-adaptive $l2$-MBQC that deterministically evaluates it.
\end{theorem}

\begin{proof}
	This follows directly from Lm.~\ref{lem: qubit delta function} and the fact that every function can be written as a sum of $\delta$-functions $f(\mathbf{i}) = \sum_{\mathbf{j} \in \zz_2^n} f_\mathbf{j} \delta(\mathbf{i} - \mathbf{j})$, $f_\mathbf{j} \in \zz_2$ for all inputs $\mathbf{i} \in \zz_2^n$,.
\end{proof}

The number of qubits in the implementation of the $\delta$-function is $N=2^{n}-1$, which is optimal (see Ref.~\cite{HobanCampbell2011}).
We explore the question of optimality for arbitrary Boolean functions in more detail in Sec.~\ref{SecGeneralResourceRequirements}, as well as other resource aspects related with $l2$-MBQC.

\section{Boolean functions as MBQC - (dependence on) Clifford hierarchy}\label{secstabiliserMBQC}

While any Boolean function can be computed using $l2$-MBQC, the type of measurements required above depended on the complexity (e.g. the degree) of the polynomial representing the Boolean function. In this section, we study the implementation of Boolean functions in $l2$-MBQC restricted to the stabiliser subtheory where only Pauli operators can be measured. In the deterministic, non-adaptive case such $l2$-MBQCs admit a simple description, namely the entire computation can be expressed as a set of eigenvalue equations that relate the inputs and outputs of the computation as follows:
\begin{equation}\label{eq: output function in ld-MBQC}
    \bigotimes_{k=1}^N U_k(c_k(\mathbf{i})) M_k(0) U_k^{-1}(c_k(\mathbf{i})) |\psi\rangle = \omega^{o(\mathbf{i})}|\psi\rangle \quad \forall \mathbf{i} \in \zz_2^n\; ,
\end{equation}
where $\omega = e^{\frac{2 \pi i}{2}} = -1$ is a square root of unity. 
In Sec.~\ref{sec: Stabiliser formalism and quadratic Boolean functions} we prove that any quadratic Boolean function can be computed within the stabiliser formalism. Conversely, any non-quadratic function requires gates from higher levels in the Clifford hierarchy. In fact, the complexity of a Boolean function in $l2$-MBQC relates to local phases via the discrete Fourier transform (see Sec.~\ref{sec: Polynomial vs Z_2-linear function representation}), which in turn puts a bound on the necessary level in the Clifford hierarchy. We make this precise in Sec.~\ref{sec: Necessity of non-Clifford operations}.

\subsection{Quadratic Boolean functions and stabiliser formalism}\label{sec: Stabiliser formalism and quadratic Boolean functions}

The qubit stabiliser formalism is contextual. For instance, the prototypical Anders-Browne NAND-gate computes a quadratic Boolean function. It is natural to ask whether stabiliser $l2$-MBQC can realise any polynomial $f: \zz_2^n \rightarrow \zz_2$. 
However, this is not the case. In fact, non-adaptive, deterministic stabiliser MBQC is limited to quadratic Boolean functions.

\begin{theorem}\label{thm:stabdeterm}
	For a non-adaptive, deterministic, level-2 (i.e., stabiliser) $l2$-MBQC only quadratic functions can be computed.
\end{theorem}

\begin{proof}
    We give the proof in App.~\ref{sec: Proof of thm:stabdeterm}
\end{proof}

For the probabilistic case, we need two additional concepts: the \emph{success probability} for a MBQC and the \emph{non-quadraticity} of a Boolean function.

\begin{definition}[success probability]\label{def: success probability}
	Let $f: \zz_2^n \rightarrow \zz_2$ be a Boolean function, and let $A$ be a MBQC, which implements $f$ with probability $p(\mathbf{i})$ on inputs $\mathbf{i} \in \zz_2^n$. We define the average success probability by $P_{succ}=\sum_{\mathbf{i} \in \zz^n_2} p(\mathbf{i})/2^n$.
\end{definition}

Assume that a deterministic MBQC implements a Boolean function $g: \zz_2^n \rightarrow \zz_2$, then the success probability is $P_{succ}=1- d_H(f,g)/2^n$ where $d_H(f,g) := |\{\mathbf{i} \in \zz_2^n \mid f(\mathbf{i}) \neq g(\mathbf{i})\}|$ denotes the Hamming distance between $f$ and $g$. 
Clearly, $P_{succ}=1$ if and only if $f=g$. In order to compute the success probability for general functions, we measure how far it is from being quadratic (see e.g.~\cite{Kolokotronis2007}).

\begin{definition}[non-quadraticity]\label{def: non-quadraticity}
    Let $f: \zz_2^n \rightarrow \zz_2$ be a Boolean function. Then the non-quadraticity of $f$ is given by
    \begin{equation}
        \mathcal{NQ}(f) := \mathrm{min} \{ d_{H}(f,q) :  q: \zz_2^n \rightarrow \zz_2\ \mathrm{quadratic}\} .
    \end{equation}
\end{definition}

It then follows as a corollary of Thm.~\ref{thm:stabdeterm} that
\begin{corollary}
\label{cor: stabiliser success probability}
	Let $f: \zz_2^n \rightarrow \zz_2$ be an arbitrary Boolean function.  The maximum success probability of computing this function in non-adaptive, deterministic, level-$2$ (i.e., stabiliser) $l2$-MBQC is
	\begin{equation}
	P_{succ} \ =\ 1 - \frac{\mathcal{NQ}(f)}{2^n} .
	\end{equation}
\end{corollary}
\begin{proof}
    The proof of this simply follows by noting that Thm.~\ref{thm:stabdeterm} entails the MBQC must compute some quadratic function $q$ with success probability $1- d_H(f,q)$. The maximum success probability is achieved by choosing $q$ to minimise the Hamming distance $d_H$, which is the non-quadraticity of $f$.
\end{proof}

In general, a non-adaptive, level-$2$ $l2$-MBQC does not yield deterministic outputs. Still, Cor.~\ref{cor: stabiliser success probability} remains true also in the probabilistic case. In other words, when restricted to stabiliser measurements (and stabiliser states), the best approximation to a given Boolean function is always achieved with a deterministic $l2$-MBQC.

\begin{theorem}\label{thm: stabiliser success probability}
	Let $f: \zz_2^n \rightarrow \zz_2$ be an arbitrary Boolean function. The maximum success probability of computing $f$ in (probabilistic) non-adaptive, level-$2$ (i.e., stabiliser) $l2$-MBQC is
	\begin{equation}
	P_{succ} \ =\ 1 - \frac{\mathcal{NQ}(f)}{2^n} .
	\end{equation}
\end{theorem}

\begin{proof}
	We give the proof in App.~\ref{sec: Proof of thm: stabiliser success probability}.
\end{proof}

In particular, this says that if $f$ is not quadratic then the success probability will be less than one, and we cannot demonstrate `strong' non-locality (contextuality) for this function. As a concrete case study, consider Example 2 in Ref.~\cite{Kolokotronis2009} that is an 8-bit input Boolean function with $\mathcal{NQ}(f)=68$. This entails an optimal success probability of $P_{succ} = 47/64 \sim 0.734375$.

Note also that the bound in Thm.~\ref{thm: stabiliser success probability} is strict since for the stabiliser formalism deterministic strategies are always optimal. However, it is not clear whether this is always the case. In particular, a similar problem arises from the well-known CHSH inequality. While quantum correlations violate the classical bound, they cannot win the related CHSH non-local game with certainty. This is different to the problem studied here, where the restriction is not on the number of qubits involved but on the level in the Clifford hierarchy of the gates used in Eq.~(\ref{eq: ld-MBQC operators by unitary conjugation}). Nevertheless, this example shows that the MBQC which best approximates a given Boolean function need not be a deterministic one.

\subsection{Beyond quadratic functions}\label{sec: Beyond quadratic functions}

Thm.~\ref{thm:stabdeterm} shows that in the non-adaptive case, non-Clifford operations are required to evaluate general (non-quadratic) Boolean functions $f$. In this section, we establish the necessity of operations belonging to higher levels in the Clifford hierarchy depending on the degree of $f$.

\subsubsection{Polynomial vs $\zz_2$-linear representation of Boolean functions}\label{sec: Polynomial vs Z_2-linear function representation}

We first introduce some tools to allow us to map between different representations of Boolean functions in $l2$-MBQC. In particular, we introduce the $\zz_2$-linear representation of a Boolean function, in addition to its polynomial representation. We show how to map between these representations using the discrete Fourier transform. This will be useful for this section and the following. (See also Ref.~\cite{o2014analysis} for more details on Boolean function analysis.) 

The polynomial representation of computational output is one useful way of characterising $l2$-MBQCs, as the polynomial degree places important constraints on the resources required. In order to characterise the optimal implementation of a given $l2$-MBQC, we consider another representation known as the $\zz_2$-linear function representation. Any Boolean function $f: \zz^n_2 \rightarrow \zz_2$ can be written in the following two ways, up to an additive constant,
\begin{equation}
f(x) = \sum_{\mathbf{a} \in \zz^n_2} C_\mathbf{a} \left( \oplus_{j=1}^n a_jx_j \right) = \sum_{\mathbf{b} \in \zz^n_2} D_\mathbf{b} \left( \prod_{j=1}^n x_j^{b_j} \right)\; ,
\end{equation}
where $C_{\mathbf{a}}\in \rr$, $D_{\mathbf{b}} \in \zz_2$ for all $\mathbf{a},\mathbf{b}\in \zz_2^n$, and $\oplus$ denotes addition modulo $2$. We focus on the first representation in terms of $\zz_2$-linear functions. In particular, we define the $\zz_2$-linear basis functions $\phi_\mathbf{a} := \oplus_{j=1}^n a_jx_j$, and monomial basis functions $\pi_{\mathbf{b}} := 2^{W(\mathbf{b})-1} \prod_{l=1}^n x_l^{b_l}$ for $\mathbf{0} \neq \mathbf{a},\mathbf{b} \in \zz^n_2$, $\phi_\mathbf{0} = \pi_\mathbf{0} := 1$, and where $W(\mathbf{b}) := |\{l \in \{1,\cdots,n\} \mid b_l \neq 0\}|$ denotes the \emph{Hamming weight} of $\mathbf{b} \in \zz^n_2$. Both sets of functions $\{\phi_{\mathbf{a}} ~|~ \mathbf{a} \in \zz_2^n\}$ and $\{\pi_{\mathbf{b}} ~|~ \mathbf{b} \in \zz_2^n\}$ are each linearly independent and generate the space of Boolean functions on bit strings $\mathbf{x} \in \zz^n_2$, as will be shown below. As such, we can determine the corresponding transformation map between the coefficients $C_\mathbf{a}$, $D_\mathbf{b}$. By Eq.~(4) in Ref.~\cite{HobanCampbell2011}, every $\zz_2$-linear function can be written in terms of monomials,
\begin{equation}\label{eq: Z2-linear functions as monomials}
\oplus_{l=1}^n x_l = \sum_{\mathbf{0} \neq \mathbf{b} \in \zz^n_2} (-2)^{W(\mathbf{b})-1} \prod_{l=1}^n x_l^{b_l}\; .
\end{equation}
Following this, it is easy to see that we can write a given $\zz_2$-linear basis function $\phi_{\mathbf{a}}$ as
\begin{equation}\label{eqLinearToPoly}
    \phi_{\mathbf{a}} = \sum_{\mathbf{0} \neq \mathbf{b} \in \mathbf{a}\zz_{2}^n} (-2)^{W(\mathbf{b})-1} \prod_{l=1}^n x_{l}^{b_l},
\end{equation}
where we have defined the set $\mathbf{a}\zz_{2}^n = \{(a_1b_1,\ldots,a_nb_n) \in \zz_2^n ~|~ \forall b_i \in \zz_2\}$. More generally, we define
the symmetric product,
\begin{equation}\label{eq: pseudo scalar product}
\langle \pi_{\mathbf{b}},\phi_\mathbf{a} \rangle := \begin{cases}
1 \quad &\text{if} ~ \mathbf{a} = \mathbf{b} = 0\; ,\\
(-1)^{\sum_{j=1}^n a_jb_j - 1} \quad &\text{otherwise}\; .
\end{cases}
\end{equation}
This defines a linear map $\mathcal{F}: \mathbb{R}^{2^n} \rightarrow \mathbb{R}^{2^n}$ with matrix coefficients $\mathcal{F}_{\pi_{\mathbf{b}}\phi_\mathbf{a}} := \langle \pi_{\mathbf{b}},\phi_\mathbf{a} \rangle$.\footnote{Note that while $\mathcal{F}$ is a map between functions over bit strings $\mathbf{i} \in \zz^n_2$ with real coefficients, it reduces to a map between Boolean functions for appropriate $\mathbf{C}_\mathbf{a}$ (and $\mathbf{D}_\mathbf{b}$). The real coefficients corresponding to a Boolean function $f$ are also known as the Walsh spectrum of $f$.} From Eq.~(\ref{eq: pseudo scalar product}) it follows that $\mathcal{F}_{\phi_\mathbf{a}\pi_{\mathbf{b}}} = \pm 1$ and given that $\mathcal{F}$ has full rank (as a basis change) it has an inverse. In fact, for fixed dimension $n$ and with appropriate normalisation factor $\mc{N} = 2^{-\frac{n}{2}}$, $\mathcal{F}$ becomes a Hadamard transform and is thus in particular orthogonal, hence, $(\mc{N}\mathcal{F}_{\pi_{\mathbf{b}}\phi_\mathbf{a}})^{-1} = \mc{N}\mathcal{F}_{\phi_\mathbf{a}\pi_{\mathbf{b}}} = \mc{N}\mathcal{F}_{\pi_{\mathbf{b}}\phi_\mathbf{a}}$. This generalises Eq.~(\ref{eq: Z2-linear functions as monomials}) and provides an explicit translation between the two representations of Boolean functions underlying Eq.~(\ref{eq: phase relations}).

More precisely, let $f = \sum_{\mathbf{a} \in \zz_2^n} C_\mathbf{a} \phi_\mathbf{a}$, where $(C_{\mathbf{a}})_{\mathbf{a}\in\zz_2^n}$ are the coefficients of $f$ in the basis $\{\phi_{\mathbf{a}}~|~\mathbf{a} \in \zz_2^n\}$, then $\mc{F}$ transforms these into coefficients $(D_{\mathbf{b}})_{\mathbf{b}\in\zz_2^n}$ in its polynomial representation $f = \sum_{\mathbf{b} \in \zz_2^n} D_\mathbf{b} (\frac{1}{2^{W(\mathbf{b})-1}} \pi_{\mathbf{b}})$, 
\begin{equation}\label{eq: map of reps on basis element}
    f = \mc{F}(\sum_{\mathbf{a} \in \zz_2^n} C_\mathbf{a} \phi_\mathbf{a})
    = \sum_{\mathbf{a} \in \zz_2^n} C_\mathbf{a} \left(\sum_{\mathbf{0} \neq \mathbf{b} \in \mathbf{a}\zz_{2}^n} (-2)^{W(\mathbf{b})-1} \prod_{l=1}^n x_{l}^{b_l}\right)
    = \sum_{\mathbf{b} \in \zz_2^n} \sum_{\mathbf{a} \in \zz_2^n} C_\mathbf{a} \langle \pi_\mathbf{b},\phi_\mathbf{a}\rangle \pi_\mathbf{b}
    = \sum_\mathbf{b} D_\mathbf{b} \pi_\mathbf{b}\; .
\end{equation}
In particular, we emphasise that the local phases $\vartheta_k$ in $\theta_k := e^{\pi i \vartheta_k}$ from Eq.~(\ref{eq: phase relations}) simply correspond to the coefficients $\mathbf{C}_\mathbf{a}$ under the mapping $\mathcal{F}^{-1}$ applied to the output function of the $l2$-MBQC (in its polynomial representation). We will apply this transformation explicitly in a number of examples in Sec.~\ref{sec: Optimal representation of functions in non-adaptive l2-MBQC} below.

\subsubsection{Necessity of non-Clifford operations}\label{sec: Necessity of non-Clifford operations}

We utilise a characterisation of the Clifford hierarchy due to Zeng et. al.~\cite{zeng2008semi}. We define a set of operations known as semi-Clifford operations~\cite{gross2007lu,zeng2008semi}.
\begin{definition}[semi-Clifford hierarchy]
    We say a gate $U\in \calC_N^k$ is a $k$-th level semi-Clifford gate (on $N$ qubits) if $U = C_1 D C_2$ where $C_1, C_2 \in \calC_N^2$ are Clifford gates, and $D \in \calC_N^k$ is diagonal. We label the set of $k$-th level semi-Clifford gates (on $N$-qubits) as $\mathcal{SC}_N^k$.
\end{definition}
In other words, gates in the semi-Clifford hierarchy are those that are diagonal up to Clifford operations. Note in the above that $D\in \calC_N^k$ necessarily, as for any $U\in \mathcal{C}_N^k$ one can verify that $C_1 U C_2 \in \mathcal{C}_N^k$ $\forall C_1, C_2 \in \mathcal{C}_N^2$~\cite{zeng2008semi}.

\begin{theorem}\label{thm:CliffHierNec}
	For a non-adaptive, deterministic $l2$-MBQC that belongs to level-$D$ in the Clifford hierarchy, only polynomials of degree at most $D$ can be computed.
\end{theorem}

\begin{proof}
    We give the proof in App.~\ref{app: Clifford hierarchy}.
\end{proof}

This theorem can be viewed as a generalisation of Thm.~\ref{thm:stabdeterm}. If a non-adaptive $l2$-MBQC belonging to some level in the Clifford hierarchy computes a polynomial of degree $D$, then it at least belongs to level-$D$ in the Clifford hierarchy. Moreover, this bound is tight, as it is saturated by the delta function construction in Lm.~\ref{lem: qubit delta function}.
	
We remark that the analogous problem for qudits is open. In our argument we used the fact that the semi-Clifford hierarchy is equal to the Clifford hierarchy for single qubits, $\mathcal{SC}_1^k = \mathcal{C}_1^k$, which has not been shown to hold for general qudits (see Ref.~\cite{de2020efficient} for a more comprehensive discussion). For prime qudits it is conjectured that all Clifford hierarchy gates are semi-Clifford, and has been proven true for the third-level gates~\cite{de2020efficient}. We also remark that certain gates do not belong to any finite level in the Clifford hierarchy. For qubits, an example is the square root of the Hadamard, $\sqrt{H}$. For qudits an example is the phase gate $D_3 = \text{diag}(1,1,-1)$.

\section{Boolean functions as MBQC - (dependence on) qubit count}\label{SecGeneralResourceRequirements}

In this section, we ask for the minimal number of qubits, also known as qubit count, needed to implement a Boolean function $f: \zz_2^n \rightarrow \zz_2$ in non-adaptive, deterministic $l2$-MBQC.

\begin{definition}
	Let $f: \zz_2^n \rightarrow \zz_2$. We call a non-adaptive $l2$-MBQC which deterministically implements $f$ \emph{optimal}, if no other non-adaptive $l2$-MBQC exists which deterministically implements $f$ on fewer qubits. The minimal number of qubits over all possible resource states is denoted by $R(f)$, while $R_\mathrm{GHZ}(f)$ denotes the minimal number of qubits when restricted to the $n$-qubit GHZ state in Eq~(\ref{eq: GHZ resource state}).
\end{definition}

Note first that we have the freedom to manipulate $f$ by any invertible linear transformation on the inputs via pre-processing $P$. The resource cost $R$ should therefore be an invariant under affine transformations.
We thus define an equivalence relation on all functions with signature $f: \zz_2^n \rightarrow \zz_2$ under affine transformations as follows,
\begin{equation}\label{eq: linear equivalence classes}
    f \sim f' :\Longleftrightarrow \exists P \in \mathrm{Mat}(n \times n, \zz_2),\ \mathrm{rk}(M) = n:\ f'(\mathbf{i}) = f(P\mathbf{i}) \quad \forall \mathbf{i} \in \zz_2^n.
\end{equation}
Furthermore, in Sec.~\ref{secCompleteness} we have seen how the $n$-dimensional $\delta$-function can be implemented as a non-adaptive $l2$-MBQC on $N=2^n-1$ qubits.\footnote{Note that this is the same scaling behaviour as for the $n$-dimensional $\mathrm{AND}(\mathbf{i}) = \prod_{j=1}^n i_j$, which is optimal by Ref.~\cite{HobanCampbell2011}.} Hence, given an arbitrary Boolean function $f: \zz_2^n \rightarrow \zz_2$, one way to implement it is by naively adding all terms in the sum $f(\mathbf{i}) = \sum_{\mathbf{j} \in \zz_2^n} f_\mathbf{j} \delta(\mathbf{i}-\mathbf{j})$ with $f_\mathbf{j} \in \zz_2$ for all $\mathbf{i} \in \zz_2^n$. However, the minimal number of qubits is only subadditive in this as well as its polynomial representation. To see this, we again consider the stabiliser case first.\\

\subsection{Qubit count in stabilzer $l2$-MBQC}\label{sec: Quadratic Boolean functions}

Recall that only quadratic functions can be computed with high probability using stabiliser $l2$-MBQCs. We now find the minimal number of qubits to do so. Consider a quadratic Boolean function
\begin{equation}\label{eq: quadratic Boolean function}
    f(x) = \sum_{i=1}^n l_i x_i + \sum_{i<j} q_{i,j} x_i x_j \pmod 2, \quad l_i, q_{i,j} \in \zz_2
\end{equation}
and define a symmetric matrix $Q(f)$ such that $Q_{i,i}=0$ and $Q_{i,j}=Q_{j,i}=q_{i,j}$. We denote by $\mathrm{rk}(f)$ the $\zz_2$-rank of $Q$.

\begin{theorem}\label{thm: resource of quadratic Boolean functions}
	Let $f: \zz_2^n \rightarrow \zz_2$ be a quadratic Boolean function. Let $f$ be expressed as 
	\begin{equation}
	f(x) = \sum_{i=1}^n l_i x_i + \sum_{i<j} q_{i,j} x_i x_j \pmod 2, \quad l_i, q_{i,j} \in \zz_2
	\end{equation}
	Then $f$ can be implemented as a non-adaptive, deterministic, level-$2$ (i.e., stabiliser) $l2$-MBQC on $R(f) = \mathrm{rk}(f)+1$ qubits, where $\mathrm{rk}(f)$ is the $\zz_2$-rank of the symmetric matrix $Q(f)$. 
\end{theorem}

\begin{proof}
	We give the proof in App.~\ref{sec: Proof of thm: resource of quadratic Boolean functions}.
\end{proof}

Thm.~\ref{thm: resource of quadratic Boolean functions} replicates the Anders-Browne result as a special case, where $o(i_1,i_2) = i_1i_2 \oplus i_1 \oplus i_2$ is quadratic with
\begin{equation}
Q = \left( \begin{array}{cc} 0 & 1 \\
1 & 0 \end{array}
\right) .
\end{equation}
This is a rank 2 matrix and so the theorem says it can be computed using 3 qubits.

Note that Thm.~\ref{thm: resource of quadratic Boolean functions} suggests another resource measure: by Thm.~\ref{thm:stabdeterm}, within the stabiliser formalism only quadratic functions can be computed; in fact, they can be computed \emph{efficiently in the number of qubits}. Unfortunately, as a consequence of Thm.~\ref{thm:stabdeterm} we cannot use arguments based on stabilisers (as in the proof of Thm.~\ref{thm: resource of quadratic Boolean functions}) to understand the number of qubits as a resource also in the general case. Instead, in the next section we will apply the Fourier transform between the polynomial and $\zz_2$-linear representation of Boolean functions to obtain a lower bound on the number of qubits for non-adaptive, deterministic $l2$-MBQCs with a GHZ resource state.\footnote{Recall that by Thm.~\ref{thm: ld-MBQC universal} GHZ-states are universal for the computation of Boolean functions in non-adaptive, deterministic $l2$-MBQC.} This turns out to be a hard problem in general, yet we show how to reproduce the bound in Thm.~\ref{thm: resource of quadratic Boolean functions}, as well as other known bounds for $R$ obtained in previous sections.

\subsection{Qubit count in $l2$-MBQC using GHZ states}\label{sec: Optimal representation of functions in non-adaptive l2-MBQC}

By comparison with optimal bounds for Bell inequalities, finding the optimal $l2$-MBQC implementing a given Boolean function is likely a difficult problem. Here, we approach this problem by fixing the resource state to be a GHZ state, which we found to be universal for non-adaptive, deterministic $l2$-MBQC in Thm.~\ref{thm: ld-MBQC universal}. We will also make use of the discrete Fourier transform defined in Eq.~(\ref{eq: map of reps on basis element}).

More precisely, let $f = \sum_{\mathbf{a} \in \zz_2^n} C_\mathbf{a} \phi_\mathbf{a}$ be a Boolean function, which is implemented (in terms of the output function) of a non-adaptive, deterministic $l2$-MBQC with a GHZ state. By Thm.~\ref{thm: phase relations ld-MBQC}, the coefficients $C_\mathbf{a}$ are encoded in terms of local phases, which in turn define local measurement operators via Eq.~(\ref{eq: l2-MBQC operators X(theta,f)}). It follows that the minimal number of qubits required to implement $f$ deterministically as a non-adaptive $l2$-MBQC with a GHZ state corresponds with the minimal number of terms in the $\zz_2$-linear representation of $f$.

Let $f = \sum_{\mathbf{b} \in \zz_2^n} D_\mathbf{b} \pi_\mathbf{b}$ be the polynomial representation of $f$. Then we obtain a corresponding representation in terms of $\zz_2$-linear functions by applying the inverse discrete Fourier transform $\mathcal{F}^{-1}$ in Eq.~(\ref{eq: map of reps on basis element}). As we will see in the next sections, for monomials and other highly symmetric functions this representation is already minimal in the number of non-zero coefficients in its $\zz_2$-linear representation, and thus in the number of qubits in the implementation as $l2$-MBQC. However, for more general Boolean functions this is no longer the case. The reason is that we may change the representation of $f$ in terms of $\zz_2$-linear functions, as long as $f$ describes the same Boolean function. To give an example, the minimal number of $\zz_2$-linear terms of the Boolean function $f: \zz_2^4 \rightarrow \zz_2$, $f(\mathbf{i}) = i_1i_2 + i_3i_4$ arises by subtracting the term $z = 4i_1i_2i_3i_4 - 2i_1i_2(i_3+i_4)$ from the `naive' representation $f(\mathbf{i}) = \frac{1}{2}(i_1 + i_2 - i_1\oplus i_2) + \frac{1}{2}(i_3 + i_4 - i_3\oplus i_4)$ given by adding the optimal representations of the Boolean functions $i_1i_2$ and $i_3i_4$.

More generally, let $f: \zz_2^n \rightarrow \zz_2$ be a Boolean function and define the linear span of \emph{zero polynomials}
\begin{equation}\label{eq: zero polynomials}
    Z(f) = \bigl\langle 2^m \sum_{\mathbf{b} \in \zz^n_2} D_\mathbf{b} \bigl( \prod_{j=1}^n x_j^{b_j} \bigr) \bigm\vert n \geq m \geq 1, D_\mathbf{b} \in \zz_2 \forall \mathbf{b} \in \zz_2^n \bigr\rangle\; .
\end{equation}
In addition to the linear equivalence relation in Eq.~(\ref{eq: linear equivalence classes}), we have the following characterisation.

\begin{theorem}\label{thm: zero term reduction}
	The minimal number of qubits $R_\mathrm{GHZ}(f)$ required to deterministically implement a given Boolean function $f: \zz_2^n \rightarrow \zz_2$ in non-adaptive $l2$-MBQC with a GHZ state, is the minimal number of non-zero coefficients $C_\mathbf{a}$ in $\mc{F}^{-1}(f)$ in Eq.~(\ref{eq: map of reps on basis element}) under the relation $f \sim f' \Longleftrightarrow f' = f + z$, $z \in Z(f)$ of Eq.~(\ref{eq: zero polynomials}).
\end{theorem}

\begin{proof}
	From the above discussion, we know that the minimal number of qubits to implement $f$ as a non-adaptive, deterministic $l2$-MBQC with a GHZ state corresponds to the minimal number of terms in the $\zz_2$-linear representation of $f$. Recall that $\mc{F}: \mathbb{R}^{2^n} \rightarrow \mathbb{R}^{2^n}$ in Eq.~(\ref{eq: map of reps on basis element}) is an orthogonal linear map, in particular, it has full rank. It follows that $\sum_\mathbf{a} C_\mathbf{a} \mc{F}(\phi_\mathbf{a}) = \mc{F}(\sum_\mathbf{a} C_\mathbf{a} \phi_\mathbf{a}) = \mathbf{0}$ for $C_\mathbf{a} \in \mathbb{R}$ implies $C_\mathbf{a} = 0$ for all $\mathbf{a} \in \zz_2^n$. Now let $f = \sum_\mathbf{a} C_\mathbf{a}\phi_\mathbf{a} = \sum_\mathbf{a} C'_\mathbf{a}\phi_\mathbf{a}$ such that $\mc{F}(\sum_\mathbf{a} C_\mathbf{a}\phi_\mathbf{a}) = \mc{F}(\sum_\mathbf{a} C'_\mathbf{a}\phi_\mathbf{a}) \pmod 2$. It follows that $\mc{F}(\sum_\mathbf{a} C_\mathbf{a}\phi_\mathbf{a}) = \mc{F}(\sum_\mathbf{a} C'_\mathbf{a}\phi_\mathbf{a}) + \mathbf{z}$ for some $\mathbf{z} \in \mathbb{R}^{2^n}$ with $\mathbf{z} = 0 \pmod 2$. Since $\{\pi_\mathbf{b}\}_{\mathbf{b} \in \zz_2^n}$ is a basis of $\mathbb{R}^{2^n}$, we conclude that $\mc{F}(\sum_\mathbf{a} C_\mathbf{a}\phi_\mathbf{a}) = \mc{F}(\sum_\mathbf{a} C'_\mathbf{a}\phi_\mathbf{a}) + z$ with $z \in Z(f)$.
	Consequently, the optimal implementation of $f$ is given by minimising the number of terms in the $\zz_2$-linear representation of $\mc{F}^{-1}(f+z)$ over all $z \in Z(f)$.\footnote{We remark that for $d>2$ the minimisation over zero polynomials in Eq.~(\ref{eq: zero polynomials}) only provides an upper bound to $R_\mathrm{GHZ}$. The reason is that the representation of the output function via $\zz_2$-linear terms in Eq.~(\ref{eq: phase relations}) breaks down for $ld$-MBQCs.}
\end{proof}

The ambiguity in the $\zz_2$-linear representation of Boolean functions makes computing the qubit count in non-adaptive, deterministic $l2$-MBQC a complex task in general. Since the number of terms in Eq.~(\ref{eq: zero polynomials}) grows doubly exponentially with $n$, a brute force search is generally infeasible. Moreover, the existence of a general solution as in the case of quadratic functions within stabiliser $l2$-MBQC via Thm.~\ref{thm:stabdeterm} seems unlikely by comparison with similar problems in circuit synthesis. For instance, the minimal number of $T$-gates can be related to minimal number of mod-$2$ linear functions with odd coefficients. Solving the latter relates to minimum distance decoding in punctured Reed-Muller codes which is hard in general \cite{AmyMosca2016,seroussi1983maximum}. 

Nevertheless, for the $\delta$-function as well as some highly symmetric functions we can use Thm.~\ref{thm: zero term reduction}, together with the discrete Fourier transform in Eq.~(\ref{eq: map of reps on basis element}), to obtain at least an estimate on the qubit count.

\subsubsection*{Example 1: $n$-dimensional $\delta$-function}

Given a general output function in its polynomial representation $o(\mathbf{i})$ we may use $\mathcal{F}^{-1}$ to obtain its representation in terms of $\zz_2$-linear basis functions and thus study its scaling behaviour. For monomials this decomposition is optimal with respect to minimising necessary $\zz_2$-linear terms.

\begin{corollary}\label{cor: optimality delta-function}
	In order to implement the monomial $f: \zz_2^n \rightarrow \zz_2$, $f(x) = \prod_{j=1}^n x_j$ in non-adaptive, deterministic $l2$-MBQC with a GHZ resource state one requires no fewer than $N = 2^n -1$ qubits, i.e., $R_\mathrm{GHZ}(f) = N$.
\end{corollary}

\begin{proof}
	Note that $f$ has degree $\mathrm{deg}(f) = n = W(\mathbf{b})$ for $\mathbf{b} = (1)^n := (1,\cdots,1) \in \zz_2^n$, hence, by Eq.~(\ref{eq: map of reps on basis element}) it has coefficient $\frac{1}{2^{W(\mathbf{b})-1}} = \frac{1}{2^{n-1}}$. Explicitly, the coefficients in the $\zz_2$-linear representation under the transformation $\mc{F}^{-1}$ read:
	\begin{align*}
	    \mathcal{F}^{-1}(\prod_{j=1}^n x_j) \stackrel{\mathbf{b} = (1)^n}{=} \mathcal{F}^{-1}(\prod_{j=1}^n x_j^{b_j}) &= \mathcal{F}^{-1}(\frac{1}{2^{W(\mathbf{b})-1}}\pi_{\mathbf{b}}) \\
	    &= \sum_{\mathbf{a} \in \zz^n_2} \frac{1}{2^{W(\mathbf{b})-1}} \langle \phi_\mathbf{a},\pi_{\mathbf{b}} \rangle \phi_\mathbf{a} = \frac{1}{2^{n-1}} \sum_{\mathbf{a} \in \zz^n_2} (-1)^{W(\mathbf{a})-1} \oplus_{j=1}^n a_j x_j
	\end{align*}
	Since these terms are all odd multiples of $\frac{1}{2^{W(\mathbf{b})-1}}$, they can only be reduced by a zero term of degree at least $n$, however, there are no such terms in $Z(f)$, hence, the representation in terms of $\zz_2$-linear functions under the transformation $\mathcal{F}^{-1}$ is already optimal. Finally, note that the overlap with $\phi_\mathbf{a}$, $\mathbf{a} = 0$ can be implemented by post-processing, leaving $N = 2^n-1$ non-zero terms. 
\end{proof}

Cor.~\ref{cor: optimality delta-function} reproduces Prop.~1 in \cite{HobanCampbell2011}. Note also that the $n$-dimensional $\delta$-function arises from monomials by linear pre-composition in Eq.~(\ref{eq: linear equivalence classes}), hence, $R_\mathrm{GHZ}(\delta) = 2^n-1$.

\subsubsection*{Example 2: elementary symmetric functions}

While for monomials the transformation in Eq.~(\ref{eq: map of reps on basis element}) is already optimal in the number of non-zero coefficients (and thus in the number of qubits in the implementation as non-adaptive, deterministic $l2$-MBQC with a GHZ resource state), this is no longer the case for more general polynomials. Nevertheless, for certain symmetric functions the minimisation problem in Thm.~\ref{thm: zero term reduction} under the equivalence relation in Eq.~(\ref{eq: zero polynomials}) simplifies.

As an example we consider elementary symmetric functions,
\begin{equation*}
    \Sigma^n_k(\mathbf{x}) \ = \sum_{\substack{i_1<\cdots<i_k\\ i_j \in \{1,\cdots,n\}}} x_{i_1} \cdots x_{i_k}, \quad \quad k \leq n\; .
\end{equation*}
Plugging $\Sigma^n_k$ into the inverse transformation in Eq.~(\ref{eq: map of reps on basis element}) results in a total number of terms equal to $\sum_{l=1}^k \binom{n}{l}$. It turns out that we can minimise this number by (at least) $\binom{n}{k}-1$ as follows. We add the zero polynomial $z \in Z(\Sigma^n_k)$ given by
\begin{align*}
    z &= (-2)^{n-k} x_1 \cdots x_n + (-2)^{n-k-1}\sum_{\substack{i_1< \cdots <i_{n-1}\\i_j \in \{1,\cdots,n\}}}x_{i_1} \cdots x_{i_{n-1}} + \cdots + (-2)\sum_{\substack{i_1< \cdots <i_{k+1}\\ i_j \in \{1,\cdots,n\}}} x_{i_1} \cdots x_{i_{k+1}} \\
    &= \sum_{l=0}^{n-k-1} (-2)^{n-k-l}\ \Sigma_l^n(\mathbf{x})\; .
\end{align*}
By construction, $\mc{F}^{-1}(\Sigma^n_k)$ and $\mc{F}^{-1}(z)$ have the same (smallest) coefficient $\frac{1}{2^{k-1}}$, and we can thus compare the coefficients in their representation based on $\zz_2$-linear functions $\phi_\mathbf{a}$, $\mathbf{a} \in \zz_2^n$. Clearly, $\mc{F}^{-1}(\Sigma^n_k+z)$ contains the term $x_1 \oplus \cdots \oplus x_n$ and thus $C^{\Sigma^n_k+z}_{W(\mathbf{a}) = n} = \frac{(-1)^{n-k}}{2^{k-1}}$. For the terms of length $k \leq m < n$, the coefficients $C^{\Sigma^n_k+z}_{W(\mathbf{a})=m}$ contain contributions from all higher degree terms in the polynomial representation of $\Sigma^n_k+z$:
\begin{align*}
    C^{\Sigma^n_k+z}_{W(\mathbf{a}) = m} &= \frac{1}{2^{k-1}} (-1)^{(n-k)+(m-1)} \left(1 - \binom{n-m}{n-m-1} + \binom{n-m}{n-m-2} - \cdots + (-1)^{n-m} \right) \\
    &= \frac{1}{2^{k-1}} (-1)^{(n-k)+(m-1)} \left( \sum_{l=0}^{n-m} (-1)^l \binom{n-m}{n-m-l}\right) = 0\; .
\end{align*}
Hence, with respect to monomials of degree $k \leq m$ in $\Sigma^n_k+z$, we have reduced the overall number of non-zero coefficients by $\binom{n}{k}-1$. Note also that the coefficients of the remaining monomials of degree $1 \leq m < k$ are non-zero since there, the above sum is truncated and reads
\begin{align*}
    C^{\Sigma^n_k+z}_{W(\mathbf{a}) = m} &= \frac{1}{2^{k-1}} (-1)^{(n-k)+(m-1)} \left(1-\binom{n-m}{n-m-1}+\binom{n-m}{n-m-2} - \cdots + (-1)^{n-k} \binom{n-m}{k-m} \right) \\
    &= \frac{1}{2^{k-1}} (-1)^{(n-k)+(m-1)} \left( \sum_{l=0}^{n-k} (-1)^l \binom{n-m}{n-m-l}\right) \neq 0\; ,
\end{align*}
thus leaving a total of $\sum_{l=1}^{k-1} \binom{n}{l} + 1$ terms in the $\zz_2$-linear representation, hence, $R_\mathrm{GHZ}(\Sigma^n_k) \leq \sum_{l=1}^{k-1} \binom{n}{l} + 1$.
Note also that: (i) $R_\mathrm{GHZ}(\Sigma^n_2) = \binom{n}{1} + 1 = n + 1$ confirms Thm.~\ref{thm: resource of quadratic Boolean functions}, since $\Sigma^n_2$ is quadratic with $\mathrm{rk}(\Sigma^n_2) = n$ (see also Prop.~2 in Ref.~\cite{HobanCampbell2011}), and (ii) $R_\mathrm{GHZ}(\Sigma^n_n) = \sum_{l=1}^{n-1} \binom{n}{l} + 1 = 2^n - 1$ reproduces the minimal number of qubits within $l2$-MBQC for monomials in Cor.~\ref{cor: optimality delta-function} (see also Prop.~1 in \cite{HobanCampbell2011}).
Comparing the latter, we draw the following conclusion from the above classification.

\begin{corollary}\label{cor: mismatch degree vs qubit count}
    There are Boolean functions $f,g: \zz_2^n \rightarrow \zz_2$ such that $\mathrm{deg}(f) > \mathrm{deg}(g)$, yet $R_\mathrm{GHZ}(f) < R_\mathrm{GHZ}(g)$.
\end{corollary}

\begin{proof}
    This follows immediately by comparing the linear scaling (in the number of qubits) of quadratic Boolean functions according to Thm.~\ref{thm: resource of quadratic Boolean functions} with the exponential scaling of the symmetric function $\Sigma^n_n$ and the $n$-qubit $\delta$-function in Cor.~\ref{cor: optimality delta-function}. For instance, $R_\mathrm{GHZ}(\Sigma^7_2) > R_\mathrm{GHZ}(\Sigma^3_3)$ despite $\mathrm{deg}(\Sigma^3_3) = 3 > 2 = \mathrm{deg}(\Sigma^7_2)$.
\end{proof}

In summary, we find that---unlike the contextuality threshold in Ref.~\cite{FrembsRobertsBartlett2018} and the close correspondence with the Clifford hierarchy in Thm.~\ref{thm:CliffHierNec}---the degree is not sufficient to compare Boolean functions with respect to their optimal representation in non-adaptive, deterministic $l2$-MBQC with a GHZ resource state. The computational classification of the latter therefore possesses a rich substructure beyond the non-contextual case.

\section{Discussion}\label{sec: discussion}

We have assessed the ability to compute Boolean functions in non-adaptive, deterministic $l2$-MBQC under various resource restrictions.
We have considered the computational power of stabiliser $l2$-MBQC, as well as $l2$-MBQC involving operations from higher levels in the Clifford hierarchy. We find that stabiliser $l2$-MBQCs can only compute quadratic functions with high probability (with the Anders and Browne example~\cite{AndersBrowne2009} being a prototypical example), while higher degree polynomials require operations from increasing levels in the Clifford hierarchy. In this way, we obtain a hierarchy of resources for non-adaptive, deterministic $l2$-MBQC beyond contextuality in \cite{Raussendorf2013}.

In addition to the necessity of certain quantum operations in $l2$-MBQC for evaluating Boolean functions, we posed the resource-theoretic problem of determining the minimal number of qubits needed to implement a given Boolean function within non-adaptive, deterministic $l2$-MBQC. Clearly, this is an important and often limiting resource for near-term quantum devices. We characterise this problem by focusing on GHZ resource states and find that it too reveals a complex substructure to contextuality. At the heart of this is the (quantum Fourier) transformation mapping between two different representations of a Boolean function, as polynomial and as a $\zz_2$-linear sum. Interestingly, our characterisation closely resembles known hard problems in circuit synthesis and minimal distance coding in punctured Reed-Muller codes \citep{AmyMosca2016}, suggesting that finding the minimal number of qubits is hard in general. Nevertheless, in certain cases the sharp bound can be found, such as for quadratic functions within stabiliser $l2$-MBQC.

Finally, we comment on some close connections and extensions of our results.

\textbf{Adaptivity.} The motivation for our setting was based on the recent results for shallow circuits, which constitute the first proof of a quantum-classical gap~\cite{BravyiGossetKoenig2018}. For this class of circuits, a constant depth circuit of one and two qubit gates is performed - that depends on the classical input bit string - followed by a measurement in the computational basis. Conversely, we consider a fixed unitary circuit (i.e., the resource state preparation), followed by a measurement that depends on the input bit string. This simplifies the analysis and allows us to derive strong bounds on resources in this scheme, but the same reasoning can also be applied in the adaptive case. As outlined in more detail in App.~\ref{secAdaptivity}, within the latter the exponential scaling in qubit count, along with the necessity of non-Clifford gates for certain functions in the non-adaptive case quickly collapse. Nevertheless, one can sometimes trade off between space and time resources such as in \cite{BravyiGossetKoenig2018}. We hope that the non-adaptive case can be leveraged to understand resource costs for more general adaptive computations.

\textbf{Magic, contextuality, and cohomology.} Both magic and contextuality can be classified by cohomology. In the former case, certain gates in the $D$'th level in the Clifford hierarchy $\mathcal{C}_N^{D}$ on $N$ qubits can be classified by elements of the group cohomology $\mathcal{H}^{D}(\zz_2^N, U(1))$, following for example Ref.~\cite{yoshida2017gapped}, while in the latter case, group cohomology also appears as a classifier for certain proofs of contextuality~\cite{AbramskyBrandenburger2011, raussendorf2016cohomological, okay2017topological, OkayTyhurstRaussendorf2018}. As both magic and contextuality appear as resources for quantum computation, it is tempting to construct a unified framework for resource theories based on cohomology. 

Recently, the role of magic in certain many-body systems known as symmetry-protected topological (SPT) phases\footnote{We remark that such phases are also classified by group cohomology~\cite{chen2013symmetry}.} has been studied~\cite{daniel2020quantum,liu2020many,ellison2020symmetry}, whereby all states within a phase of matter possess magic. Such SPT phases have also been identified as resources for MBQC~\cite{else2012symmetry,NWMBQC,miller2016,raussendorf2019computationally,devakul2018universal,roberts2020symmetry}. It would be interesting to study the role of many-body magic for computational universality, particularly with the example of Ref.~\cite{miller2016}, which is universal with only Pauli measurements. Further, it would be interesting to consider the role of contextuality in the fault-tolerant setting -- particularly fault-tolerant MBQC~\cite{raussendorf2005long,raussendorf2006fault,raussendorf2007fault,raussendorf2007topological,brown2020universal,bravyi2020quantum} -- where non-Clifford operations require vastly more resources than non-Clifford operations (and indeed is the motivation for considering magic as a resource in the present setting).

\begin{acknowledgements} 
We acknowledge support from Australian Research Council via the Centre of Excellence in Engineered Quantum Systems (EQUS) project number CE170100009. MF was supported through a studentship in the Centre for Doctoral Training on Controlled Quantum Dynamics at Imperial College funded by the EPSRC, as well as through grant number FQXi-RFP-1807 from the Foundational Questions Institute and Fetzer Franklin Fund, a donor advised fund of Silicon Valley Community Foundation, and ARC Future Fellowship FT180100317. ETC's technical contributions were made while at the University of Sheffield.  SDB acknowledges additional support from the Australian Research Council via project number DP220101771.
\end{acknowledgements}

\bibliographystyle{unsrtnat}
\bibliography{literature.bib}
\appendix

\section{Proof of Theorem~\ref{thm: phase relations ld-MBQC}}\label{app: proof of Thm 1}
We first prove Thm.~\ref{thm: phase relations ld-MBQC} for the local measurement operators $X(\theta)$, $\theta = e^{i\pi \vartheta}$ in Eq.~(\ref{eq: l2-MBQC operators X(theta,f)}). The relation between their eigenstates and the computational basis reads as follows:
\begin{equation*}
    |m\rangle_{\vartheta} = \frac{1}{\sqrt{2}}(|0\rangle + (-1)^m e^{i\pi\vartheta}|1\rangle)\; .
\end{equation*}
Conversely, the computational basis expressed in terms of eigenstates of $X(\theta)$ reads
\begin{equation}\label{eq: computational basis state in measurement basis2}
    |q\rangle = \frac{1}{\sqrt{2}}e^{-i\pi q\vartheta}\sum_{m=0}^{1} (-1)^{qm} \ket{m}_{\vartheta}\; .
\end{equation}
We encode the choice of local measurement operators by $M_k(c_k(\mathbf{i})) = X(e^{i\pi c_k(\mathbf{i}) \vartheta_k})$ for linear functions $c_k: \zz_2^n \rightarrow \zz_2$ with $c_k(0) = 0 \pmod 2$\footnote{In other words, $c_k = \phi_{\mathbf{a}_k}$ with $0 \neq \mathbf{a}_k \in \zz_2^n$ for every $k \in \{1,\cdots,N\}$ (see Sec.~\ref{sec: Polynomial vs Z_2-linear function representation}).} and parameters $\vartheta_k \in [0,1)$. In particular, note that $M(0) = X$.  Rewriting the $N$-qubit GHZ resource state in Eq.~(\ref{eq: GHZ resource state}) in terms eigenstates of the local measurement bases thus yields
\begin{align}
    |\psi\rangle = \frac{1}{\sqrt{2}} \sum_{q=0}^1 |q\rangle^{\otimes N}
    &= \frac{1}{\sqrt{2}} \sum_{q=0}^1 \otimes_{k=1}^N \left(\frac{1}{\sqrt{2}} e^{-i\pi q c_k(\mathbf{i})\vartheta_k} \sum_{m_k=0}^1 (-1)^{qm_k} |m_k\rangle_{\vartheta} \right) \nonumber\\
    &= \left(\frac{1}{\sqrt{2}}\right)^{N+1} \sum_{q=0}^1 \left( \sum_{\mathbf{m} \in \zz_2^N} (-1)^{q(\sum_{k=1}^N m_k - o'(\mathbf{i}))} \otimes_{k=1}^N |m_k\rangle_{\vartheta} \right) \nonumber\\
    &= \left(\frac{1}{\sqrt{2}}\right)^{N-1} \left( \sum_{\substack{\mathbf{m} \in \zz_2^N, \\ \oplus_{k=1}^N m_k = o'(\mathbf{i})}} \otimes_{k=1}^N|m_k\rangle_{\vartheta} \right)\; , \label{eq: parity eigenstate}
\end{align}
where we defined $(-1)^{o'(\mathbf{i})} = e^{-i\pi q c_k(\mathbf{i})\vartheta_k}$ and we used that $|\psi\rangle$ is an eigenstate of the global measurement operators $M(\mathbf{i}) = \otimes_{k=1}^N M_k(c_k(\mathbf{i}))$. Finally, since the output function of the non-adaptive, deterministic $l2$-MBQC reads $o(\mathbf{i}) = \oplus_{k=1}^N m_k$, we find $o'=o$, hence,
\begin{equation*}
    o(\mathbf{i}) = \sum_{k=1}^N c_k(\mathbf{i})\vartheta_k \pmod{2} \; .
\end{equation*}

We are left to show that every local measurement operator is of the form in Eq.~(\ref{eq: l2-MBQC operators X(theta,f)}). To see this, note that the global measurement operators $M(\mathbf{i}) = \otimes_{k=1}^N M_k(c_k(\mathbf{i}))$ are such that $|\psi\rangle$ is a parity eigenstate of $M(\mathbf{i})$ for all inputs $\mathbf{i} \in \zz_2^n$. 
For every $\mathbf{i} \in \zz_2^n$, rewrite $|\psi\rangle$ in the local eigenbases corresponding to the $M_k(c_k(\mathbf{i}))$. This yields a superposition of product states $|\mathbf{m}\rangle_{\varphi,\vartheta} = \otimes_{k=1}^N |m_k\rangle_{\varphi,\vartheta}$, where we again denote every product state by the Boolean vector $\mathbf{m} \in \zz_2^N$ such that
\begin{equation*}
    |0\rangle_{\varphi,\vartheta} = \sin(\varphi)|0\rangle + e^{\pi i\vartheta} \cos(\varphi)|1\rangle \quad \quad \quad 
    |1\rangle_{\varphi,\vartheta} = \cos(\varphi)|0\rangle - e^{\pi i\vartheta} \sin(\varphi)|1\rangle\; .
\end{equation*}
In particular, note that $|m\rangle_{\vartheta} = |m\rangle_{\frac{\pi}{4},\vartheta}$. Clearly, the product state $|\mathbf{m}\rangle_{\varphi,\vartheta}$ has parity $m := \oplus_{k=1}^N m_k$. Moreover, the coefficient to the product state $|\mathbf{m}\rangle_{\varphi,\vartheta}$ reads
\begin{equation}\label{eq: product state coefficients}
    {}_{\varphi,\vartheta}\langle \mathbf{m} \vert \psi \rangle = \frac{1}{\sqrt{2}} \left(\prod_{k=1}^N \Phi^{m_k}(\varphi_k) + (-1)^m e^{\pi i \sum_{k=1}^N \vartheta_k } \prod_{k=1}^N \Phi^{m_k \oplus 1}(\varphi_k)\right)\; ,
\end{equation}
where we defined $\Phi^0(\varphi_k) = \sin(\varphi_k)$ and $\Phi^1(\varphi_k) = \cos(\varphi_k)$, and the two summands correspond to the inner product between the two summands in $|\psi\rangle = \frac{1}{\sqrt{2}}(|0\rangle^{\otimes N} + |1\rangle^{\otimes N})$ with $|\mathbf{m}\rangle_{\varphi,\vartheta}$.\footnote{Note that local measurements in the computational basis only change the resource state and can thus be neglected.}

We have a parity eigenstate if ${}_{\varphi,\vartheta}\langle \mathbf{m} \vert \psi \rangle=0$ for all $\mathbf{m}$ with $m \neq o$ for some $o \in \zz_2$. 
We thus obtain $\frac{2^N}{2}$ constraints from Eq.~(\ref{eq: product state coefficients}), both on absolute values and phases of the form
\begin{equation}\label{eq: product state coefficients2}
    \prod_{k=1}^N \Phi^{m_k}(\varphi_k) + (-1)^m e^{\pi i \sum_{k=1}^N\vartheta_k}  \prod_{k=1}^N \Phi^{m_k \oplus 1}(\varphi_k) = 0 \quad \quad \quad \forall \mathbf{m} \in \zz_2^N \text{ s.t. } m \neq o\; .
\end{equation}
Clearly, the constraints on absolute values are satisfied for $\varphi = \frac{\pi}{4}$. Moreover, for $N \geq 3$ all solutions are of this form. First, for $N \geq 3$ odd, consider pairs of constraints in Eq.~(\ref{eq: product state coefficients}) of the same parity $m = \oplus_{k=1}^N m_k$. Specifically, given any $\mathbf{m} \in \zz_2^N$ and another vector arising from $\mathbf{m}$ by flipping all bits except the one at site $k$. Then we have the following pair of constraints,
\begin{align*}
    \Phi^{m_k}(\varphi_k)\prod_{k'\neq k} \Phi^{m_{k'}}(\varphi_{k'}) + (-1)^me^{\pi i \sum_{k=1}^N\vartheta_k} \Phi^{m_k \oplus 1}(\varphi_k)\prod_{k'\neq k} \Phi^{m_{k'} \oplus 1}(\varphi_{k'}) &= 0 \\
    \Phi^{m_k}(\varphi_k)\prod_{k'\neq k} \Phi^{m_{k'} \oplus 1}(\varphi_{k'}) + (-1)^me^{\pi i \sum_{k=1}^N\vartheta_k} \Phi^{m_k \oplus 1}(\varphi_k)\prod_{k'\neq k} \Phi^{m_{k'}}(\varphi_{k'}) &= 0
\end{align*}
These imply $\frac{\prod_{k'\neq k}\Phi^{m_{k'}}(\varphi_{k'})}{\prod_{k'\neq k}\Phi^{m_{k'} \oplus 1}(\varphi_{k'})} = (-1)^{m+1}e^{\pi i \sum_{k=1}^N\vartheta_k} \frac{\Phi^{m_k \oplus 1}(\varphi_k)}{\Phi^{m_k}(\varphi_k)} = \frac{\prod_{k'\neq k}\Phi^{m_{k'} \oplus 1}(\varphi_{k'})}{\prod_{k'\neq k}\Phi^{m_{k'}}(\varphi_{k'})}$ and thus $|\sin(\varphi_k)| = |\cos(\varphi_k)|$, hence, $\varphi_k = \frac{\pi}{4}$. For $N$ even, similar constraints yield $|\Phi^{m_k}(\varphi_k)\Phi^{m_{k'}}(\varphi_{k'})| = |\Phi^{m_k \oplus 1}(\varphi_k)\Phi^{m_{k'} \oplus 1}(\varphi_{k'})|$. For $N \neq 2$ we thus again find $\varphi_k = \frac{\pi}{4}$, since for another pair of constraints in Eq.~(\ref{eq: product state coefficients}) also $|\Phi^{m_k}(\varphi_k)\Phi^{m_{k'} \oplus 1}(\varphi_{k'})| = |\Phi^{m_k+ \oplus 1}(\varphi_k)\Phi^{m_{k'}}(\varphi_{k'})|$, hence, $\frac{|\Phi^{m_{k'}}(\varphi_{k'})|}{|\Phi^{m_{k'} \oplus 1}(\varphi_{k'})|} = \frac{|\Phi^{m_k \oplus 1}(\varphi_k)|}{|\Phi^{m_k}(\varphi_k)|} = \frac{|\Phi^{m_{k'} \oplus 1}(\varphi_{k'})|}{|\Phi^{m_{k'}}(\varphi_{k'})|}$. Finally, for all $N > 2$ we find $(-1)^{m+1} e^{\pi i \sum_{k=1}^N \vartheta_k}=1$, hence, $e^{\pi i \sum_{k=1}^N \vartheta_k} = (-1)^{m+1} = (-1)^o$, which recovers the first part of the proof.

\section{Proof of Lemma~\ref{lem: qubit delta function}}\label{app: proof of Lm 1}

Consider the resource state given by the $N$-qubit GHZ state $|\psi\rangle$ in Eq.~(\ref{eq: GHZ resource state}) with $N=2^n-1$, and consider the measurement procedure $0 \rightarrow M(0) = X$ and $1 \rightarrow M(1) = X(\theta) = X(e^{i\pi\vartheta})$ with measurements in Eq.~(\ref{eq: l2-MBQC operators X(theta,f)}), which we re-state here for convenience,
\begin{equation*}\label{eq: generalised qubit operators}
    X(\theta) = \left( \begin{array}{cc}
       0 & \theta^*  \\
    \theta & 0  
    \end{array} \right) \quad \quad X(\theta)|q\rangle = \theta^{1-2q}|q\oplus 1\rangle = e^{i\pi(1-2q)\vartheta}|q\oplus 1\rangle\; .
\end{equation*}
As before, the measurement operators $M_k(c_k(\mathbf{i})) = X(e^{i\pi c_k(\mathbf{i}) \vartheta_k})$ are specified by linear functions $c_k: \zz_2^n \rightarrow \zz_2$. In particular, we set
\begin{equation*}
    c_k(\mathbf{i}) := \phi_\mathbf{a}(\mathbf{i}) = \oplus_{j=1}^n a_j i_j\; , \quad \quad \mathbf{0} \neq \mathbf{a} \in \zz_2^n\; .
\end{equation*}
In other words, the qubits in $|\psi\rangle$ are indexed by vectors $\mathbf{0}\neq \mathbf{a} \in \zz_2^n$. We prove that this indeed allows us to compute the $n$-dimensional $\delta$-function in Eq.~(\ref{eq: n-dimensional delta function}) for a suitable $\vartheta_k = \vartheta$.

First, consider the case of the input string containing exactly one non-zero entry, e.g. $\mathbf{i}^T = (1,0,\cdots,0)$, and count the number of phases $\vartheta$ that we collect. As $\vartheta$ is independent of the site, this is simply the number of functions $\phi_\mathbf{a}$ that $i_1$ appears in. There is one function in which it appears by itself, then $n-1$ functions where it appears together with another input, $\binom{n-1}{2}$ in which it appears together with two more inputs and so on. Overall, the number of functions is
\begin{equation*}
    \sum_{k=0}^{n-1} \binom{n-1}{k} = 2^{n-1}\; .
\end{equation*}
For inputs containing two non-zero entries, e.g., $\mathbf{i} = (1,1,0,\cdots,0)$, we again count the number of appearances of, in this case, $i_1$ and $i_2$. Note that only those functions $\phi_\mathbf{a}$ will contribute for which $\phi_\mathbf{a}(\mathbf{i}) = 1$, i.e., those that contain exactly one but not both of entries. The corresponding counting of appearances is thus given by
\begin{equation*}
\sum_{k=0}^{n-2} \binom{2}{1}\binom{n-2}{k} = 2 \cdot 2^{n-2} = 2^{n-1}.
\end{equation*}
The general case with $m$ non-zero entries reads as follows:
\begin{equation*}
\sum_{k=1}^{\left \lceil{\frac{m}{2}}\right \rceil} \binom{m}{2k-1} \sum_{l=0}^{n-m} \binom{n-m}{l} = 2^{m-1} \cdot 2^{n-m} = 2^{n-1}
\end{equation*}
Hence, for all but the zero input we flip the overall parity in Eq.~(\ref{eq: phase relations}) if we set
\begin{equation}\label{eq: qubit delta function phase}
(e^{\pi i \vartheta})^{2^{n-1}} = -1\ \Longleftrightarrow\ \vartheta = 2^{-(n-1)} \pmod 2\; . 
\end{equation}
Finally, note our setup computes the function $o(\mathbf{i}) = \delta(\mathbf{i}) + 1$, hence, we obtain the $n$-dimensional $\delta$-function by simple post-processing.

\section{Universality of non-adaptive, deterministic $ld$-MBQC with $d$ prime}\label{sec: Proof of lm: qudit delta function}

In this section we show that any function $f: \zz_d^n \rightarrow \zz_d$ for $d$ prime, can be implemented using a non-adaptive, deterministic $ld$-MBQC. We will follow a similar strategy to the proof of Thm.~\ref{thm: ld-MBQC universal} in App.~\ref{app: proof of Thm 1} We start by choosing measurement operators for prime dimension $d$ similar to those in  in  Eq.~(\ref{eq: l2-MBQC operators X(theta,f)}), namely
\begin{equation*}
    M(0)|q\rangle := X|q\rangle = |q\oplus 1\rangle, \quad \quad
    M(c)|q\rangle
    := \theta(c) \chi^{c q^{d-1}} |q\oplus 1\rangle
    \quad 1 \leq c \leq d-1\; .
\end{equation*}
If we set $\theta(c)^d = \chi^{-(d-1)c}$ we have $M(c)^d = 1$ for all $c \in \zz_d$. We find the following eigenstates,
\begin{align*}
    |m\rangle_{\theta(c)} &= \frac{1}{\sqrt{d}}(|0\rangle + \omega^{-m} \theta(c)|1\rangle + (\omega^{-m}\theta(c))^2\chi^c|2\rangle + \cdots + (\omega^{-m}\theta(c))^{d-1}(\chi^c)^{d-2}|d-1\rangle) \\
    &= \frac{1}{\sqrt{d}}\sum_{q=0}^{d-1} (\omega^{-m}\theta(c))^{q}(\chi^c)^{(q-1)q^{d-1}}|q\rangle\; ,
\end{align*}
with corresponding expressions in terms of computational basis states,
\begin{equation}\label{eq: computational basis state in measurement basis}
    |q\rangle = \frac{1}{\sqrt{d}}\frac{1}{\theta(c)^{q}\chi^{c(q-1)q^{d-1}}} \sum_{m=0}^{d-1} \omega^{qm} |m\rangle_{\theta(c)}, \quad \forall c \in \zz_d\; .
\end{equation}
We fix the resource state to be the $N$-qudit GHZ state for $N = d^n-1$,
\begin{equation*}
    |\psi\rangle = \frac{1}{\sqrt{d}}\sum_{q=0}^{d-1}|q\rangle^{\otimes N}\; .
\end{equation*}
Assume that the output function $o: \zz^n_d \rightarrow \zz_d$ is encoded in the phase relations as follows
\begin{equation}\label{eq: parity phase relations1}
    \prod_{1 \leq k \leq d^n-1} \theta(c_k(\mathbf{i}))^q\chi^{c_k(\mathbf{i})(q-1)q^{d-1}} = \omega^{qo(\mathbf{i})}\; .
\end{equation}
Rewriting $|\psi\rangle$ in terms of the respective measurement bases via Eq.~(\ref{eq: computational basis state in measurement basis}) then yields
\begin{align*}
    |\psi\rangle &= \frac{1}{\sqrt{d}} \sum_{q=0}^{d-1} \otimes_{k=1}^N \left(\frac{1}{\sqrt{d}} \frac{1}{\theta(c_k(\mathbf{i}))^{q}\chi^{c_k(\mathbf{i})(q-1)q^{d-1}}} \sum_{m_k=0}^{d-1} \omega^{qm_k} |m_k\rangle_{\theta(c_k)} \right) \\
    &= \left(\frac{1}{\sqrt{d}}\right)^{N+1} \sum_{q=0}^{d-1} \left( \omega^{-qo(\mathbf{i})} \sum_{\mathbf{m} \in \zz_d^N} \otimes_{k=1}^N \omega^{qm_k} |m_k\rangle_{\theta(c_k)} \right) \\
    &= \left(\frac{1}{\sqrt{d}}\right)^{N+1} \sum_{q=0}^{d-1} \left( \sum_{\mathbf{m} \in \zz_d^N} \omega^{q(\sum_{k=1}^N m_k - o(\mathbf{i}))} \otimes_{k=1}^N |m_k\rangle_{\theta(c_k)} \right) \\
    &= \left(\frac{1}{\sqrt{d}}\right)^{N-1} \sum_{\substack{\mathbf{m} \in \zz_d^N, \\ \sum_{k=1}^N m_k = o(\mathbf{i}) \pmod d}} \otimes_{k=1}^N|m_k\rangle_{\theta(c_k)}\; .
\end{align*}
It follows that $|\psi\rangle$ is a parity $\sum_{k=1}^N m_k = o(\mathbf{i}) \pmod d$ eigenstate for all operators $M(\mathbf{i})$ with $\mathbf{i} \in \zz_d^n$.

We thus want to show that we can satisfy the phase relations in Eq.~(\ref{eq: parity phase relations1}) for any $o:\zz_d^n \rightarrow \zz_d$ by choosing suitable linear functions for the measurement settings $c_k$. Similar to the case of Boolean functions in Sec.~\ref{sec: Polynomial vs Z_2-linear function representation}, we take as a basis for the space of functions $f: \zz^n_d \rightarrow \zz_d$ all (non-zero) linear functions of the form $\phi_\mathbf{a} = \oplus_{j=1}^n a_j i_j$ for $\mathbf{0} \neq \mathbf{a} \in \zz_d^n$. In analogy with the proof of Lm.~\ref{lem: qubit delta function} in App.~\ref{app: proof of Lm 1} we again count vectors $\mathbf{0} \neq \mathbf{a} \in \zz_d^n$ such that $\phi_\mathbf{a}(\mathbf{i}) \neq 0$ for input $\mathbf{i} \in \zz_d^n$.

First, consider a single non-zero entry, $\zz_d^n \ni \mathbf{i} = (i_1,0,\cdots,0)$ and let $\mathbf{a}$ such that $a_1 \neq 0$. There is $(d-1)$-fold degeneracy resulting from changing $\mathbf{a}$ to $\mathbf{a}'$ such that $a_1' = ra_1$ for some $0\neq r \in \zz_d$ and $a_j' = a_j$ for all $j >2$. This degeneracy yields a the local phase factor\footnote{Note that we are abusing notation slightly by using modulo-$d$ arithmetic over phases with different periods. However, as the functions are computed classically the input is always an element in $\zz_d$.}
\begin{equation}\label{eq: local phase}
    \phi(q) := \prod_{c=0}^{d-1} \theta(c)^q\chi^{c(q-1)q^{d-1}} = \theta^q \chi^{\sum_{c=0}^{d-1} c(q-1)q^{d-1}}
    = \theta^q \chi^{\frac{d(d-1)}{2}(q-1)q^{d-1}},
\end{equation}
where we set $\theta := \prod_{c=0}^{d-1} \theta(c)$. Furthermore, the number of functions $\phi_\mathbf{a}$ with $a_1 \neq 0$ counts $\sum_{k=0}^{n-1} \binom{n-1}{k} (d-1)^k = d^{n-1}$, hence, the overall phase factor in Eq.~(\ref{eq: parity phase relations1}) reads $\phi(q)^{d^{n-1}}$.

Next we consider an input with two non-zero entries, e.g., $\mathbf{i}=(i_1,i_2,0,\cdots,0)$. We need to be more careful about the counting in this case as in contrast to case $\zz_2$, where $i_1 \oplus i_2 = 0$ for $i_1,i_2 \neq 0$, this does not hold in $\zz_d$. For functions $\phi_\mathbf{a}$ with non-zero coefficients $a_1,a_2 \neq 0$, we instead need to count how many combinations $a_1i_1 + a_2i_2 \neq 0 \pmod d$ for any $i_1,i_2 \neq 0$ fixed. It is not hard to see that there are $(d-1)^2-(d-1) = (d-1)(d-2)$ non-trivial combinations, hence, we end up with the following global phase factor in Eq.~(\ref{eq: parity phase relations1}),
\begin{equation*}
    (\phi(q))^{d^{n-2}}
    (\phi(q))^{d^{n-2}}
    (\phi(q)^{d-2})^{d^{n-2}} = (\phi(q)^d)^{d^{n-2}} = \phi(q)^{d^{n-1}}.
\end{equation*}
The first two factors count functions $\phi_\mathbf{a}$ where either $a_1 =0$ or $a_2 = 0$. The third arises from functions $\phi_\mathbf{a}$ with both $a_1,a_2 \neq 0$, out of which there are $((d-1)(d-2))^{d^{n-2}}$ (and where by symmetry we can always group $(d-1)$ together to obtain the phase $\phi(q)$ from Eq.~(\ref{eq: local phase})).

This argument generalises to inputs $\mathbf{i} \in \zz_d^n$ with $m$ non-zero entries. Note first that the number of non-zero linear combinations in $\phi_\mathbf{a} = \oplus_{j=1}^n a_ji_j$ with $W(\mathbf{a}) = k$, denoted by $g(k)$, is given by,
\begin{align*}
    g(1) = (d-1), \quad \quad g(k) &= (d-1)^k - g(k-1) \\
    &= (d-1)^k - (d-1)^{k-1} + g(k-2) \\
    &\hspace{1cm} \vdots \\
    &= \sum_{l=0}^{k-1} (-1)^l (d-1)^{k-l}.
\end{align*}
Now there are $d^{n-m}$ functions for every function that contains at least one $a_k \neq 0$ and for each of those we have the contribution,
\begin{align*}
    \sum_{k=1}^m \binom{m}{k}g(k)
    &= \sum_{k=1}^m \binom{m}{k} \left( \sum_{l=0}^{k-1} (-1)^l (d-1)^{k-l} \right) \\
    &= \binom{m}{m} \left( (d-1)^m - (d-1)^{m-1} + (d-1)^{m-2} - (d-1)^{m-3} + \cdots \right) \\
    &\hspace{1.2cm} +\binom{m}{m-1} \left( (d-1)^{m-1} - (d-1)^{m-2} + (d-1)^{m-3} - \cdots \right) \\
    &\hspace{3.3cm} +\binom{m}{m-2} \left( (d-1)^{m-2} - (d-1)^{m-3} + \cdots \right) \\
    &\hspace{7cm} \vdots \\
    &= \sum_{k=0}^{m-1} (-1)^k (d-1)^{m-k} \left( \sum_{l=0}^k (-1)^l \binom{m}{m-l} \right) \\
    &= \sum_{k=0}^{m-1} (-1)^k (d-1)^{m-k} \left( (-1)^k \binom{m-1}{k} \right) \\
    &= (d-1)\left( \sum_{k=0}^{m-1} (d-1)^k \binom{m-1}{k} \right) \\
    &= (d-1)d^{m-1}.
\end{align*}
In here, the first factor $(d-1)$ will result in the phase $\phi(q)$ from Eq.~(\ref{eq: local phase}) and we thus again obtain the global phase,
\begin{equation}
    (\phi(q)^{d^{m-1}})^{d^{n-m}} = \phi(q)^{d^{n-1}}.
\end{equation}
Finally, we relate this global phase factor to the local phases $\theta(c)$ and $\chi$. Since,
\begin{equation}\label{eq: phase constraint 1}
    \theta^d = \left(\prod_{c=1}^{d-1} \theta(c)\right)^d = \prod_{c=1}^{d-1} \chi^{-(d-1)c} = \chi^{-(d-1)\sum_{k=1}^{d-1} c} = \chi^{-\frac{d(d-1)^2}{2}},
\end{equation}
we need to choose $\theta(c)$ for $1 \leq c \leq d-1$ such that $\theta = \chi^{-\frac{(d-1)^2}{2}}$, e.g. $\theta(c) := \chi^{-\frac{c(d-1)}{d}}$. Next, we insert Eq.~(\ref{eq: phase constraint 1}) into Eq.~(\ref{eq: local phase}) and compute the global phase factors,
\begin{equation*}
    \phi(q)^{d^{n-1}} = (\chi^{-\frac{(d-1)^2}{2}q} \cdot \chi^{(q-1)\frac{d(d-1)}{2}q^{d-1}})^{d^{n-1}}
    = \begin{cases} 1 &\mathrm{if\ } q=0 \\ 
    \chi^{-\frac{d^{n-1}d(d-1)}{2}} (\chi^{-d^{n-1}\frac{(d-1)}{2}})^q &\mathrm{if\ } 1 \leq q \leq d-1\\ \end{cases}.
\end{equation*}
We may thus set $\chi^{-\frac{d^{n-1}(d-1)}{2}} = \omega$ from which it follows that $(\chi^{-\frac{d^{n-1}(d-1)}{2}})^d = 1$, hence, $\phi(q)^{d^{n-1}} = \omega^{q}$. We obtain the following output function,
\begin{equation}
    o(\mathbf{i}) = \begin{cases} 0 &\mathrm{if} \ \mathbf{i} = 0 \\
    1 &\mathrm{if} \ \mathbf{i} \neq 0 \end{cases}\; ,
\end{equation}
from which we compute $\delta(\mathbf{i}) = (d-1) o(\mathbf{i}) + 1$ by simple post-processing.

Finally, as in Thm.~\ref{thm: ld-MBQC universal} the result follows since every function can be written as a sum of $\delta$-functions, $f(\mathbf{i}) = \sum_{\mathbf{j} \in \zz_d^n} f_\mathbf{j} \delta(\mathbf{i} - \mathbf{j})$, $f_\mathbf{j} \in \zz_d$ for all inputs $\mathbf{i} \in \zz_d^n$.

\section{Proof of Theorem \ref{thm:stabdeterm}}\label{sec: Proof of thm:stabdeterm}

Recall from Def.~\ref{defldMBQC} that a non-adaptive, deterministic, level-$2$ (i.e., stabiliser) $l2$-MBQC is specified by the following data: $P \in \mathrm{Mat}(N\times n,\zz_2)$ is the classical, linear pre-processing; $|\psi\rangle$ is an $N$-qubit stabiliser resource state; and $M_k(c_k(\mathbf{i}))$ is a single qubit Pauli operator for every $c_k(\mathbf{i})  = P\mathbf{i}$, input $\mathbf{i} \in \zz_2^n$ and $k \in \{1,\cdots,N\}$. 
Given a stabiliser $l2$-MBQC as above, note that the exact same measurement statistics are obtained if we instead rotate $|\psi\rangle$ by some local Clifford operations and conjugate the measurement settings by the inverse Clifford operations. Therefore, we can assume without loss of generality that $M_k(0)=X_k$ and $M_k(1)=Z_k$. We denote $M(\mathbf{c} = P\mathbf{i})$ to be the tensor product of unitaries measured in this canonical choice.

Consider the quadratic function $f(\mathbf{x}) = \sum_{i=1}^n l_i x_i + \sum_{i<j} q_{i,j} x_i x_j$ with associated matrix $Q$. We require
\begin{align}
    M(\mathbf{0}) \ket{\psi}& = (+1) \ket{\psi} \\
M(P \mathbf{x}) \ket{\psi}& = (- 1)^{f(\mathbf{x})} \ket{\psi} \quad \forall \mathbf{0} \neq \mathbf{x} \in \zz_2^n\; .
\end{align}
If we denote by $\mathcal{S}$ the stabiliser of $\ket{\psi}$, this can be restated as $M(\mathbf{0}) \in \mathcal{S}$ and $(-1)^{f(\mathbf{x})}M(P\mathbf{x}) \in \mathcal{S}$ for all $\mathbf{0} \neq \mathbf{x} \in \zz_2^n$.  Under group closure of the stabiliser we have $(-1)^{f(\mathbf{x})}M(P\mathbf{x}) M(\mathbf{0}) \in \mathcal{S}$.  Note that $M(\mathbf{0})=X^{\otimes N}$ and $M(P\mathbf{x})$ will be a tensor product of $X$ and $Z$ operators. Therefore, the product $M(P\mathbf{x}) M(\mathbf{0})$ is a tensor product of the identity and $Y$ operators, possibly with some extra phase. We define 
\begin{align}
    Q(\mathbf{u}) := \bigotimes_{k=1}^N (i Y)^{u_k} = i^{W(\mathbf{u})} \bigotimes_{k=1}^N  Y^{u_k}  ,
\end{align}
where we recall that $W(\mathbf{u}) := |\{k \in \{1,\cdots,N\} \mid u_k \neq 0\}|$ denotes the \emph{Hamming weight} of $\mathbf{u} \in \zz_2^N$, and we observe that $M(P\mathbf{x}) M(\mathbf{0})=Q(P\mathbf{x})$. Therefore, $(-1)^{f(\mathbf{x})}Q(P\mathbf{x}) \in \mathcal{S}$ for all $\mathbf{0} \neq \mathbf{x} \in \zz_2^n$. In particular, for $\mathbf{x}=(1,0,\ldots)^T$ and $\mathbf{x}=(0,1,\ldots)^T$ we have $(-1)^{l_1}Q(\mathbf{p}_1) \in \mathcal{S}$, $(-1)^{l_2}Q(\mathbf{p}_2) \in \mathcal{S}$, where we write $\mathbf{p}_j$ to denote the $j^{\mathrm{th}}$ column of $P$.

Assuming the stabiliser is abelian, $Q(P\mathbf{x})$ ought to be Hermitian and so $W(P\mathbf{x})=0 \pmod{2}$ for all $\mathbf{x} \in \zz_2^n$. Next we note that we have the relation
\begin{equation}
\label{eq_relations}
    Q(\mathbf{u}) Q(\mathbf{v}) = (-1)^{\mathbf{u} \cdot \mathbf{v}} Q(\mathbf{u} \oplus \mathbf{v})\; .
\end{equation}
Since  $(-1)^{l_1}Q(\mathbf{p}_1) \in \mathcal{S}$ and $(-1)^{l_2}Q(\mathbf{p}_2) \in \mathcal{S}$, we have $(-1)^{l_1+l_2}Q(\mathbf{p}_1)Q(\mathbf{p}_2) \in \mathcal{S}$ by group closure. Using the above relation, this entails that $(-1)^{l_1 + l_2 + \mathbf{p}_1 . \mathbf{p}_2}Q(\mathbf{p}_1 \oplus \mathbf{p}_2) \in \mathcal{S}$, where $\mathbf{p}_1 . \mathbf{p}_2$ is the dot product of these vectors. However, we also know that  $(-1)^{f(1,1,\ldots )} Q(P(1,1,0, \ldots, 0 )^T)\in \mathcal{S}$. These two results are only compatible if 
\begin{equation}
    (-1)^{l_1 + l_2 + \mathbf{p}_1 . \mathbf{p}_2}=(-1)^{l_1 + l_2 + q_{1,2}}
\end{equation}
and so $\mathbf{p}_1 . \mathbf{p}_2 = q_{1,2} \pmod{2}$. A similar argument shows that for all $i,j \in \{1,\cdots,n\}$ we must have 
\begin{align}\label{Eq_contraints}
	\mathbf{p}_i . \mathbf{p}_j & = q_{i,j} \pmod{2} \\ \nonumber
	W(\mathbf{p}_i)& =0 \pmod{2}\; .
\end{align}

We have so far checked inputs $\mathbf{x} \in \zz_2^n$ with Hamming weight $W(\mathbf{x}) \leq 2$. More generally, let $\mathbf{x} \in \zz_2^n$ be arbitrary such that $P \mathbf{x} = \bigoplus_{i=1}^n \mathbf{p}_ix_i$.
For every $i \in \{1,\cdots,n\}$ with $x_i =1$, we find $(-1)^{l_i}Q(\mathbf{p}_i)\in \mathcal{S}$ as before, hence, by group closure
\begin{equation} \label{Step1}
    \prod_{i=1}^n (-1)^{l_ix_i}Q(\mathbf{p}_ix_i) =(-1)^{\sum_{i=1}^n l_i x_i} \prod_{i=1}^n Q(\mathbf{p}_ix_i) \in \mathcal{S}\; .
\end{equation}
Repeated application of Eq.~(\ref{eq_relations}) then yields
\begin{align}
    \prod_{i=1}^n Q(\mathbf{p}_ix_i)
    & = (-1)^{\sum_{i < j} \mathbf{p}_i . \mathbf{p}_j x_ix_j} Q(\bigoplus_{i=1}^n \mathbf{p}_ix_i)
    \nonumber \\
    & = (-1)^{\sum_{i < j} q_{i,j} x_i x_j }   Q(P\mathbf{x}) \label{Step2}\; ,
\end{align}
where in the second line we have used $\mathbf{p}_i. \mathbf{p}_j= q_{i,j}$ from Eq.~(\ref{Eq_contraints}). 
Combining Eqs.~(\ref{Step1}) and (\ref{Step2}) gives
\begin{equation} 
    (-1)^{\sum_{i=1}^n l_i x_i + \sum_{i<j} q_{i,j} x_i x_j}Q(P \mathbf{x}) = (-1)^{f(\mathbf{x})}Q(P\mathbf{x}) \in \mathcal{S}.
\end{equation}
This proves that any quadratic function can be computed within non-adaptive, deterministic, level-$2$ $l2$-MBQC. Conversely, for any Boolean function $f: \zz_2^n \rightarrow \zz_2$ the above argument shows that only its quadratic part can be computed deterministically. Hence, $f$ can be computed by a non-adaptive, deterministic, level-$2$ $l2$-MBQC if and only if $f$ is quadratic.

\section{Proof of Theorem \ref{thm: stabiliser success probability}}\label{sec: Proof of thm: stabiliser success probability}

In this section, we prove Thm.~\ref{thm: stabiliser success probability}, which bounds the success probability of non-adaptive, level-$2$ (i.e., stabiliser) $l2$-MBQC. Let $f:\zz_2^n \rightarrow \zz_2$ be a Boolean function. Then the closest Boolean function (in Hamming distance) which can be deterministically computed in non-adaptive, level-$2$ $l2$-MBQC is a quadratic function. Hence, the success probability is determined by the non-quadraticity of $f$ if we restrict to deterministic $l2$-MBQC (recall  Cor.~\ref{cor: stabiliser success probability}). However, it is not immediately clear that a deterministic $l2$-MBQC necessarily performs best, i.e., it maximises the success probability. Here we show that for non-adaptive, level-$2$ $l2$-MBQC this is indeed the case. 

Let $A$ be a non-adaptive, level-$2$ $l2$-MBQC that given $\mathbf{x} \in \zz_2^n$, outputs $f(\mathbf{x})$ with probability $p_A(\mathbf{x})$ so that
\begin{equation}
    P_{\mathrm{succ}}(A) = \frac{1}{2^n} \sum_{\mathbf{x} \in \zz_2^n} p_A(\mathbf{x}).
\end{equation}
If $A$ is probabilistic, we let $D_A := \{\mathbf{x} \in \zz_2^n \mid p_A(\mathbf{x})\in \{0,1\}\}$ denote the subset of values such that the outcome is deterministic. We denote the complement by $R_A := \zz_2^n\backslash D_A$, which is the random subset on which  $0<p_A(\mathbf{x})<1$.  If $R_A$ is empty, $A$ has deterministic outcomes and we can deploy Cor.~\ref{cor: stabiliser success probability}. We will show that when $R_A$ is not empty, we can find a deterministic (non-adaptive, level-$2$) $l2$-MBQC $A^\star$ with $P_{\mathrm{succ}}(A^\star) \geq P_{\mathrm{succ}}(A)$.

\begin{lemma} \label{5050Lem}
	For all $\mathbf{x} \in R_A$, $p_A(\mathbf{x})=1/2$.
\end{lemma}
\begin{proof}
For every $\mathbf{x}\in \zz_2^n$, let $S(\mathbf{x})$ be the observable measured.  Assuming the stabiliser state has stabiliser $\mathcal{S}$, there are two possible cases, either
\begin{enumerate}
	\item $S(\mathbf{x}) \in \mathcal{S}$ or $-S(\mathbf{x}) \in \mathcal{S}$ and so $p_A(\mathbf{x}) \in \{ 0, 1\}$ and $\mathbf{x} \in D_A$;
	\item or $S(\mathbf{x})$ anti-commutes with some element in $\mathcal{S}$ in which case $p_A(\mathbf{x})=1/2$ and $\mathbf{x} \in R_A$.
\end{enumerate}	
This proves the lemma.
\end{proof}

From Lm.~\ref{5050Lem} it follows that
\begin{equation} \label{FormEq}
	\sum_{\mathbf{x} \in \zz_2^n} p_A(\mathbf{x})  =  	   \frac{1}{2}|R_A| + \sum_{\mathbf{x} \in D_A}  p_A(\mathbf{x} )  = \frac{1}{2}(2^n - 2^m ) + \sum_{\mathbf{x} \in D_A}  p_A(\mathbf{x})\; , 
\end{equation}
where we have used that $|R_A|=2^n-|D_A|=:2^n-2^m$.

\begin{lemma}\label{lm: affine space}
    For all $\mathbf{x},\mathbf{y},\mathbf{z} \in D_A$, $\mathbf{x} \oplus \mathbf{y} \oplus \mathbf{z} \in D_A$.
\end{lemma}	 
\begin{proof}
Consider the measurement $S(\mathbf{x})$.  W.l.o.g we can assume it has the form
\begin{equation}
	S(\mathbf{x}) = \otimes_k X_k (iZ_k)^{[P\mathbf{x}]_k}\; ,
\end{equation}	
where $P$ is the matrix describing the ($\zz_2$-linear) pre-processing (see Def.~\ref{defldMBQC}). From this we find that
\begin{equation}
	\label{Sxyz}
	S(\mathbf{x}) S(\mathbf{y}) S(\mathbf{z}) 
	\propto \otimes_j X_j (iZ_j)^{[P(\mathbf{x} \oplus \mathbf{y} \oplus \mathbf{z})]_j}
	= S(\mathbf{x} \oplus \mathbf{y} \oplus \mathbf{z})\; ,
\end{equation}	
where the proportionality constant can be worked out but is not important. Assuming $\mathbf{x},\mathbf{y},\mathbf{z} \in D_A$ entails that $S(\mathbf{x}), S(\mathbf{y})$ and $S(\mathbf{z})$ all commute with $\mathcal{S}$. Therefore, $S(\mathbf{x}) S(\mathbf{y}) S(\mathbf{z})$ must also commute with $\mathcal{S}$, and by Eq.~\eqref{Sxyz} we know $S(\mathbf{x} \oplus \mathbf{y} \oplus \mathbf{z})$ must also commute with $\mathcal{S}$.  Therefore, $\mathbf{x} \oplus \mathbf{y} \oplus \mathbf{z} \in D_A$.
\end{proof}

We remark that this is the structure of an affine space. Recall that an affine space is a set $\{ \mathbf{y} \oplus \mathbf{w} : \mathbf{y} \in L \}$ where $L$ is a linear space and $\mathbf{w}$ is some constant shift.
Let $\mathbf{w} \in D_A$ arbitrary, and define $L_A^\mathbf{w} =\{ \mathbf{x} \oplus \mathbf{w} : \mathbf{x} \in D_A \}$. The space $L_A^\mathbf{w}$  is linear: from $\mathbf{x} \in L_A^\mathbf{w} \Rightarrow (\mathbf{x} \oplus \mathbf{w}) \in D_A$ and $\mathbf{y} \in L_A^\mathbf{w} \Rightarrow (\mathbf{y} \oplus \mathbf{w}) \in D_A$ it follows that $\mathbf{x}\oplus \mathbf{y} \in L_A^\mathbf{w}$ since, by Lm.~\ref{lm: affine space}, $(\mathbf{x} \oplus \mathbf{w}) \oplus (\mathbf{y} \oplus \mathbf{w}) \oplus \mathbf{w} = (\mathbf{x} \oplus \mathbf{y} )\oplus \mathbf{w} \in D_A$. Hence, $D_A$ is an affine space.

Since $D_A$ is an affine space, we can define an invertible, affine transformation $\Phi:\zz_2^n \rightarrow \zz_2^n $ such that the image $\Phi(D_A)$ corresponds to the vectors of the form $(u_1,\ldots,u_m,0,\ldots,0) \in \zz_2^n$.  It is convenient to change the problem under this transformation, in particular, we define the new target function by $g(\Phi(\mathbf{x}))=f(\mathbf{x})$. We also define the truncated function $\tilde{g}:\zz_2^m \rightarrow \zz_2$ such that $\tilde{g}(\mathbf{u})=g(u_1,\ldots,u_m,0,\ldots,0)$. 
 
Since $A$ is deterministic over $D_A$, by Thm.~\ref{thm:stabdeterm} it defines a quadratic Boolean function $\tilde{q}_A: \zz_2^m \rightarrow \zz_2$ on inputs in $\Phi(D_A)$.
Clearly, the success probability (with respect to the different target functions $f$ and $g$) is invariant under the transformation $\Phi$ (being a mere relabelling of inputs), hence, Eq.~\eqref{FormEq} becomes
\begin{align} \label{TransFormEq}
    \sum_{\mathbf{x} \in \zz_2^n} p_A(\mathbf{x}) = \sum_{\mathbf{x} \in \zz_2^n} p_A(\Phi(\mathbf{x})) =  \frac{1}{2}(2^n - 2^m )+ (2^m -d_H(\tilde{q}_A , \tilde{g} )) 
    =  \frac{1}{2}(2^n + 2^m ) -d_H(\tilde{q}_A , \tilde{g} )\; ,
\end{align}
where we recall that $d_H(\tilde{q}_A,\tilde{g}) := |\{\mathbf{x} \in \zz_2^m \mid \tilde{q}_A(\mathbf{x}) \neq \tilde{g}(\mathbf{x})\}|$ denotes the Hamming distance between $\tilde{q}_A$ and $\tilde{g}$. Next, we extend $\tilde{q}_A$ to a quadratic function on all inputs $\mathbf{x} \in \zz_2^n$.

\begin{lemma}\label{EmbeddedFunctionLem}
    Let $\tilde{g}$ be a Boolean function $\tilde{g}:\zz_2^m \rightarrow \zz_2$ with an extension $g:\zz_2^m \times \zz_2^{n-m} \rightarrow \zz_2 $.  For any quadratic function $\tilde{q}_A:\zz_2^m \rightarrow \zz_2$, we can find a quadratic function $q_A:\zz_2^m \times \zz_2^{n-m} \rightarrow \zz_2 $ such that
    \begin{equation}
 	    d_H(q_A, g ) \leq \frac{1}{2}(2^n - 2^m ) + d_H(\tilde{q}_A, \tilde{g} )\; .
    \end{equation}	
\end{lemma}
\begin{proof}
    The proof is recursive.  We define the series of nested extension functions $g^{(j)}:  \zz_2^{{m+j}} \rightarrow \zz_2 $ such that  $g^{(0)}=\tilde{g}$ and $g^{(n-m)}=g$, where  $g^{(j)}(\mathbf{u})=g^{(j+1)}(\mathbf{u},0)$ for all $\mathbf{u} \in \zz_2^{m+j}$. We will recursively define a series of quadratic functions $q_A^{(j)}:  \zz_2^{m+j} \rightarrow \zz_2 $ starting with $q_A^{(0)}=\tilde{q}_A$, such that $q_A^{(j)}(\mathbf{u})=q_A^{(j+1)}(\mathbf{u},0)$ and $\Delta_j + q_A^{(j)}(\mathbf{u})=q_A^{(j+1)}(\mathbf{u},1)$ for all $\mathbf{u} \in \zz_2^{m+j}$ for some constant $\Delta_j \in \zz_2$ to be determined. Clearly, the $q_A^{(j)}$ are all quadratic if and only if $q_A^{(0)}$ is quadratic. Furthermore,
    \begin{align}
        d_H(q_A^{(j+1)}, g^{(j+1)} ) & =  \sum_{\mathbf{u} \in \zz_2^{m+j}} \left[ q_A^{(j+1)}(\mathbf{u},0) \oplus g^{(j+1)}(\mathbf{u},0)   \right]  +  \sum_{\mathbf{u} \in \zz_2^{m+j}} \left[ q_A^{(j+1)}(\mathbf{u},1) \oplus g^{(j+1)}(\mathbf{u},1)   \right]  \\ \nonumber
        & =  d_H(q_A^{(j)}, g^{(j)} ) +  \sum_{\mathbf{u} \in \zz_2^{m+j}} \left[ \Delta_j \oplus q_A^{(j)}(\mathbf{u}) \oplus g^{(j+1)}(\mathbf{u},1)  \right]\; ,
    \end{align}	
    Assume the sum in the last line evaluates to $N$ when $\Delta_j=0$, then it evaluates to $2^{m+j}-N$ when  $\Delta_j=1$.  Therefore, we can choose $\Delta_j$ such that the sum evaluates to $2^{m+j}/2$ or less. This yields
    \begin{align}
    	d_H(q_A^{(j+1)}, g^{(j+1)} ) & \leq  d_H(q_A^{(j)}, g^{(j)} ) + 2^{m+j}/2\; .
    \end{align}
    Using our initial condition for $j=0$, and applying this bound recursively we get
    \begin{align}
    	d_H(q_A, g ) =	d_H(q_A^{(n-m)}, g^{(n-m)} )
    	\leq  d_H(q_A^{(0)}, g^{(0)} ) +  \frac{1}{2}\sum_{j=0}^{n-m-1} 2^{m+j}
    	= d_H( \tilde{q}_A, \tilde{g} ) +   \frac{1}{2}(2^n - 2^m)\; ,
    \end{align}
    which proves the lemma.
\end{proof}

Since $q_A$ from Lm.~\ref{EmbeddedFunctionLem} is quadratic, by Thm.~\ref{thm:stabdeterm} we can find a non-adaptive, deterministic, level-$2$ $l2$-MBQC $A^{\star}$ (using stabiliser states) implementing $q_A$. Applying Lm.~\ref{EmbeddedFunctionLem} and comparing with Eq.~(\ref{TransFormEq}) we obtain
\begin{align} 
    \sum_{\mathbf{x} \in \zz_2^n} p_{A^\star}(\mathbf{x})  & =  2^n -	d_H(q_A, g )
    \geq 2^n -  d_H( \tilde{q}_A, \tilde{g} ) -   \frac{1}{2}(2^n - 2^m )
    = \frac{1}{2}(2^n + 2^m ) -  d_H( \tilde{q}_A, \tilde{g} )
    =  \sum_{\mathbf{x} \in \zz_2^n} p_A(\mathbf{x})\; ,
\end{align}
such that the deterministic $l2$-MBQC $A^{\star}$ performs at least as well as probabilistic $l2$-MBQC $A$. Thm.~\ref{thm: stabiliser success probability} thus follows from Cor.~\ref{cor: stabiliser success probability}.

\section{Proof of Theorem \ref{thm:CliffHierNec}}\label{app: Clifford hierarchy}

Let the $l2$-MBQC belong to level-D in the Clifford hierarchy (per Def.~\ref{defLevelDMBQC}). Then each measurement $M_k(c_k)$ takes the form 
\begin{equation}
    M_k(c_k) = U_k({c_k}) M_k(0) U_k^{\dagger}({c_k})\; ,
\end{equation}
where $U_k({c_k}) \in \mathcal{C}_1^D$ and $c_k: \zz_2^n \rightarrow \zz_2$ a linear function for all $k \in \{1,\ldots,N\}$ (see Def.~\ref{defldMBQC}). For deterministic computation we have for each input $\mathbf{i}\in \zz_2^n$
\begin{equation}
    \bigotimes_{k=1}^N \left[ U_k(c_k) M_k(0) U_k^{\dagger}(c_k)\right] \ket{\psi} = \omega^{o(\mathbf{i})}\ket{\psi}
\end{equation}
where $\ket{\psi}$ is the resource state and $o(\mathbf{i})$ is the computational output. From Ref.~\cite{zeng2008semi}, we have $\mathcal{SC}_1^D  = \mathcal{C}_1^D$, meaning $\exists C_k, C_k' \in \mathcal{C}_1^2$ and diagonal gates $D_k \in \mathcal{C}_1^D$ such that $U_k = C_k D_k C_k'$. Eq.~(\ref{eq: output function in ld-MBQC}) can be rewritten as 
\begin{equation}\label{eqOutputWithClifford}
\bigotimes_{k=1}^N \left[ D_k(c_k) \tilde{M}_k(0) D_k^{\dagger}(c_k) \right] \ket{\tilde \psi} = \omega^{o(\mathbf{i})}\ket{\tilde \psi},
\end{equation}
where $\ket{\tilde{\psi}} = C_k^{\dagger} \ket{\psi}$ and $\tilde{M}(0) = \bigotimes_{k=1}^N (C_k' M_k(0) C_k'^{\dagger})$. Note that $\ket{\tilde \psi}$ is a stabiliser state and $\tilde{M}(0)$ is a Pauli operator, which we write 
\begin{equation}\label{eqFidMeasTilde}
    \tilde{M}(0) = e^{\frac{i\pi\beta}{2}}(X_1^{x_1} Z_1^{z_1}) \otimes \ldots \otimes (X_N^{x_N} Z_N^{z_N}),
\end{equation}
for $\beta\in \zz_4$ and $\mathbf{x} = (x_1,\ldots,x_N) \in \zz_2^N$, $\mathbf{z} = (z_1,\ldots,z_N) \in \zz_2^N$. Expanding $\ket{\tilde \psi}$ in the computational basis, we have 
\begin{equation}
    \ket{\tilde \psi} = \sum_{\mathbf{q} \in \zz_2^N} \alpha(\mathbf{q}) \ket{\mathbf{q}} \quad \text{ such that } \quad  |\alpha(\mathbf{q})| \in \{0, \alpha\} \ \forall \mathbf{q} \in \zz_2^N, \text{ for some } \alpha \in \rr,
\end{equation}
which follows the fact that all nonzero amplitudes of a stabilizer state in the computational basis have the same magnitude. The global measurements in the updated basis
\begin{equation}
    \tilde{M}(\mathbf{c}) = \bigotimes_{k=1}^N D_k(c_k) \tilde{M}_k(0) D_k^{\dagger}(c_k)
\end{equation}
permute computational basis states up to a phase,
\begin{equation}\label{eq: action of global measurement operators}
    \tilde{M}(\mathbf{c}) \ket{\mathbf{q}} = \theta({\mathbf{c}, \mathbf{q}}) \ket{\mathbf{q} \oplus \mathbf{x}},
\end{equation}
where $\theta({\mathbf{c}, \mathbf{q}}) \in U(1)$ for all $\mathbf{c},\mathbf{q} \in \zz_2^N$. To satisfy Eq.~(\ref{eqOutputWithClifford}) we must have 
\begin{equation}
    \theta({\mathbf{c}, \mathbf{q}}) = \omega^{o(\mathbf{i})} \quad \forall \mathbf{q}\in \zz_2^N \text{ with } \alpha(\mathbf{q}) \neq 0, ~\forall \mathbf{c} \in \zz_2^N.
\end{equation}
Thus, the dependence on $\mathbf{q}$ may be dropped and we may write $\theta({\mathbf{c}})  := \theta({\mathbf{c}, \mathbf{q}})$, and we remark that $\mathbf{c}$ is implicitly dependent on the input $\mathbf{i}$.

To determine the allowable phases $\theta({\mathbf{c}}) = \omega^{o(\mathbf{i})}$, we utilise a classification of diagonal gates in the Clifford hierarchy from Ref.~\cite{cui2017diagonal}. For any function $f: \zz_2^N\rightarrow U(1)$ we denote by $D[f]$ the diagonal operator whose action is given by $D[f]\ket{\mathbf{q}} = f(\mathbf{q}) \ket{\mathbf{q}}$ for all $\mathbf{q} \in 
\zz_2^N$. From Ref.~\cite{cui2017diagonal}, up to a global phase every diagonal function $D \in \mathcal{C}_1^D$ can be written as $D[f]$ where $f:\zz_2 \rightarrow U(1)$ is given by
\begin{equation}
    f(q_k) = \exp\left(2\pi i\sum_{m=0}^D \frac{\vartheta_m q_k}{2^m} \right), \quad \text{for some } \vartheta_m \in \zz_{2^m} \quad \forall q_k\in \zz_2.
\end{equation}

Then each $\tilde{M}_k(c_k) = D_k(c_k)\tilde{M}_k(0) D_k^{\dagger}(c_k)$, with $\tilde{M}_k(0) = e^{\frac{i\pi \beta_k}{2}} X_k^{x_k}Z_k^{z_k}$ has an action on computational basis states as
\begin{equation}
    \tilde{M}_k(c_k) \ket{q_k} = \exp\left[2\pi i\left(\frac{\beta_k}{4} + \frac{z_k q_k}{2} + \sum_{m=0}^{D} \frac{\vartheta_{m,k}[c_k] x_k}{2^m} \right) \right] \ket{q_k \oplus x_k}, \quad \text{for some } \vartheta_{m,k}[c_k] \in \zz_{2^m}, \quad \forall q_k \in \zz_2.
\end{equation}
Therein, the factors $\vartheta_{m,k}[c_k]$ are determined by the choice of gate $D_k(c_k) = \text{diag}(1, \exp\left(2\pi i\sum_{m=0}^D \frac{\vartheta_{m,k}[c_k]}{2^m} \right))$. In particular, we may rewrite them as $\vartheta_{m,k}[c_k] = \vartheta_{m,k}[0](1-c_k) + \vartheta_{m,k}[1]c_k$, for $\vartheta_{m,k}[0]$, $\vartheta_{m,k}[1] \in \zz_{2^m}$. 

Then the global phase, and thus computational output can be obtained by accumulating all local phases,
\begin{equation}
    \tilde{M}(\mathbf{c})\ket{\mathbf{q}} = \exp\left[2\pi i\left( \frac{\beta}{4} + \sum_{k=1}^{N} \frac{z_k q_k}{2} + \sum_{m=0}^{D}\sum_{k=1}^N \frac{ \vartheta_{m,k}[0](1 - c_k) + \vartheta_{m,k}[1]c_k x_k}{2^m} \right) \right] \ket{\mathbf{q} \oplus \mathbf{x}}, \quad \forall \mathbf{q}, \mathbf{x} \in \zz_2^N.
\end{equation}

Equating the phase in the above expression to $\omega^{o(\mathbf{i})}$ as dictated by Eq.~(\ref{eq: action of global measurement operators}) we have
\begin{equation}\label{eqPhaseFunction}
    o(\mathbf{i}) = \frac{\beta}{2} + \sum_{k=1}^{N} {z_k q_k} + \sum_{m=0}^{D}\sum_{k=1}^N \frac{ \vartheta_{m,k}[0](1 - c_k) + \vartheta_{m,k}[1]c_k x_k}{2^{m-1}} \quad \pmod{2}.
\end{equation}

Now we recall that the measurement settings may be written as $\zz_2$-linear basis functions $c_k = \phi_{\mathbf{a}}$ for some $\mathbf{a}_k\in \zz_2^n$ (where $\phi_{\mathbf{a}}$ is defined in Sec.~\ref{sec: Polynomial vs Z_2-linear function representation}). Then using Eq.~(\ref{eqLinearToPoly}) we rewrite this function in the monomial basis
\begin{equation}\label{eqLinearToPoly2}
    c_k = \sum_{\mathbf{0} \neq \mathbf{b} \in \mathbf{a}_k\zz_{2}^n} (-2)^{W(\mathbf{b})-1} \prod_{l=1}^n i_{l}^{b_l}.
\end{equation}
Inserting into Eq.~(\ref{eqPhaseFunction}), we conclude that the third term in Eq.~(\ref{eqPhaseFunction}) contributes only if $m \geq W(\mathbf{b})$. Moreover, since $D \geq m$, and since the degree of the monomial term in Eq.~(\ref{eqLinearToPoly2}) is given by $W(\mathbf{b})$, any non-vanishing term in the output function in Eq.~(\ref{eqPhaseFunction}) has degree at most $D$. This completes the proof.

\section{Proof of Theorem \ref{thm: resource of quadratic Boolean functions}}\label{sec: Proof of thm: resource of quadratic Boolean functions}

Following the terminology of the proof of Thm.~\ref{thm:stabdeterm} in App.~\ref{sec: Proof of thm:stabdeterm}, we denote by $P \in \mathrm{Mat}(N\times n,\zz_2)$ the classical, $\zz_2$-linear pre-processing of a non-adaptive, deterministic, level-$2$ (i.e., stabiliser) $l2$-MBQC. 

It follows that the qubit count equals the number of rows in $P$, hence, we seek a suitable $P$ with minimal number of rows. We also recall the conditions $\mathbf{p}_i . \mathbf{p}_j = q_{i,j} \pmod{2}$ and $W(\mathbf{p}_i)=0 \pmod{2}$ (see App.~\ref{sec: Proof of thm:stabdeterm}) for any non-adaptive, deterministic, level-$2$ $l2$-MBQC computing the quadratic function $f$. The latter constraints are equivalent to $Q(f)=P^TP \pmod{2}$, where $Q(f)$ is the matrix associated with $f$ in Eq.~(\ref{eq: quadratic Boolean function}). It was shown by Lempel~\cite{Lempel1975} that a solution $P$ always exists and that the smallest number of rows of $P$ equals $N = \mathrm{rk}(Q(f)) + 1$. This completes the proof.

\section{Adaptivity}\label{secAdaptivity}
In this section, we comment further on how our results change in the presence of adaptive measurements. Adaptivity is a powerful resource for many quantum computational schemes. For universal MBQC it is essential---in general, measurement bases must be chosen based on previous measurement outcomes in order to control the randomness induced by non-deterministic measurement outcomes. For many families of quantum circuits adaptivity is also essential and they may become classically simulable in its absence, see Ref.~\cite{jozsa2013classical} for example.

By conditioning future measurements on prior measurement outcomes, qubit count and non-Clifford resource requirements can be drastically reduced. To see this, we consider a general adaptive MBQC as being composed of several non-adaptive MBQCs called components (where each component does need not to have deterministic output). The overall computation can be represented by a directed acyclic graph $\calG$ called the incidence graph. Each node on the graph $\calG$ corresponds to a non-adaptive component, and the directed edges correspond to the information flow required for adaptivity: the target node corresponds to the component that requires the output of the component corresponding to the source node. 

The nodes of the graph $\calG$ are also labelled by integers, referring to the order in which they are performed. Multiple nodes may share the same label -- meaning they are performed in parallel -- but the labels must strictly increase when moving along the edges. We call this list of integers the schedule $\calS$. For a MBQC with incidence graph $\calG$ and schedule $\calS$, we define the depth of the computation as the largest integer in $\calS$. We define the width of the computation as the total number of qubits in all components with a common schedule index $k\in \calS$, maximised over all $k\in \calS$.

\begin{figure}[h]%
	\centering
	\includegraphics[width=0.5\linewidth]{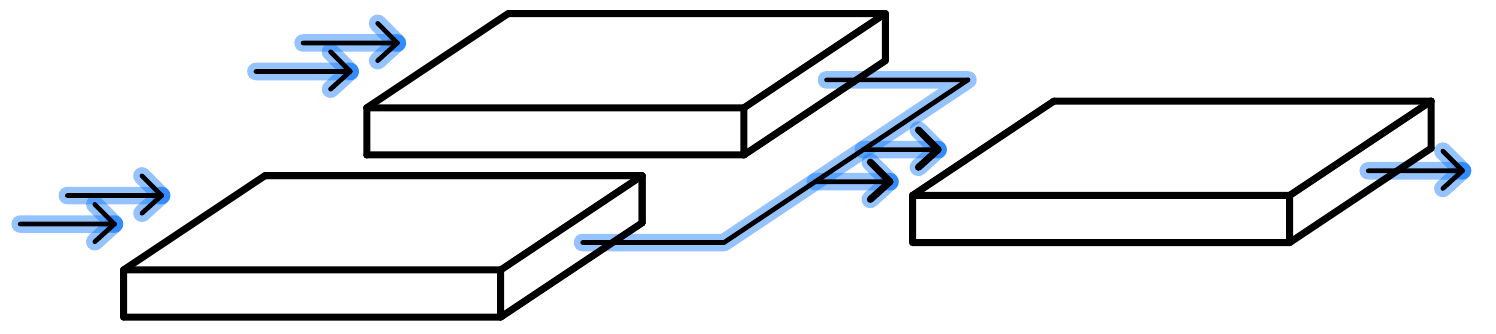}
	\caption{Sequence of non-adaptive MBQCs composed together.}%
\end{figure}

The depth is how many timesteps the computation takes to perform, while the width is how many qubits are required to execute it with the prescribed schedule. Note the volume does not represent the number of qubits required to implement the MBQC. In fact, an arbitrary width $w$ MBQC can be implemented using $w$ qubits as not all measurements need to be executed in parallel. In general, space-time tradeoffs are possible, meaning that it may be possible to vary between the width and the depth of the MBQC. We note that shallow quantum circuits in \cite{BravyiGossetKoenig2018} are restricted to constant depth, while non-adaptive MBQCs admit a depth-1 representation.

As a concrete example, we consider the delta function $\delta: \zz_2^n \rightarrow \zz_2$. As shown in Cor~\ref{cor: optimality delta-function}, $2^n-1$ qubits are required for its implementation in non-adaptive MBQC (i.e. width $2^n-1$ and depth 1), as well as non-Clifford gates belonging to the $n$-th level in the Clifford hierarchy. Using the adaptive scheme represented in Fig.~\ref{figBinaryTreeScheme} (left), the delta function can be implemented with width $3$ (meaning only 3 qubits are required), however the depth needed is $n$. The volume of $3n$ is exponentially smaller (in $n$) than the non-adaptive case. Similarly, one could choose an adaptive scheme based on a binary tree, such as that depicted in Fig.~\ref{figBinaryTreeScheme} (right). In this case, one can use $O(3n)$ qubits and a depth of $O(\log(n))$ to compute the delta function. This gives a volume of $O(3n\log(n))$. \\

\begin{figure}[h]%
\centering
\includegraphics[width=0.95\linewidth]{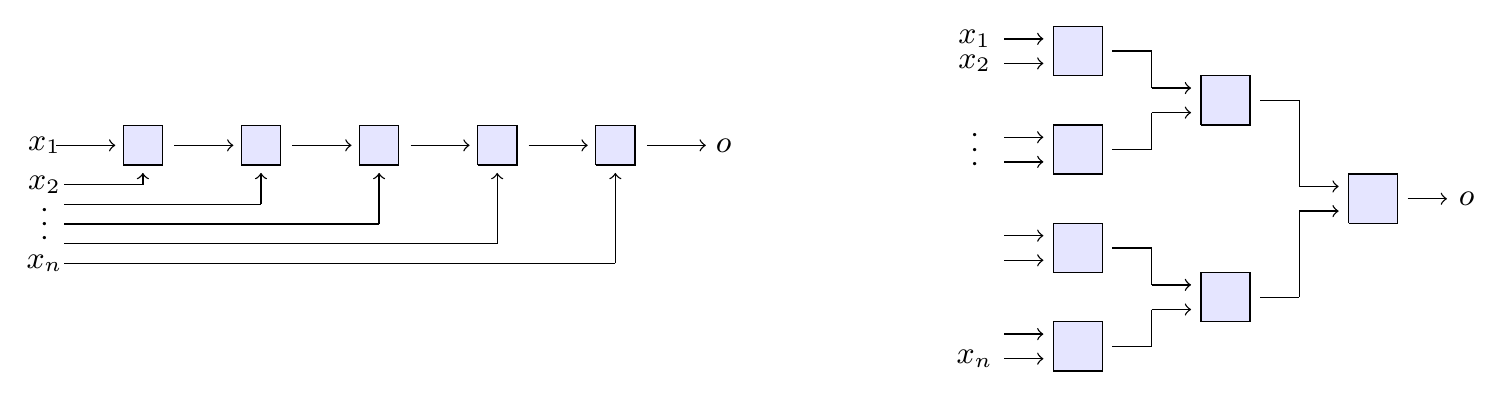}
\caption{Linearly composed (left) and binary tree composition (right) of Anders and Browne MBQCs to compute the delta function. Each box represents an $l2$-MBQC such that the output is the product of the two inputs.}%
\label{figBinaryTreeScheme}
\end{figure}

\end{document}